\documentclass[12pt,a4paper]{scrartcl}
\usepackage[utf8]{inputenc}
\usepackage[english]{babel} 
\usepackage{csquotes}
\usepackage{amsmath}
\usepackage{amsfonts}
\usepackage{amssymb}
\usepackage{graphicx}
\usepackage[margin=2cm, footskip=1cm]{geometry}

\usepackage[scr=rsfso, cal=cm]{mathalfa} 

\usepackage{titling}

\title{
	The Proca Field in Curved Spacetimes and its\\Zero Mass Limit
	}
\author{Maximilian Schambach}
\date{\today}

\usepackage{titletoc}
\usepackage{scrtime}
\usepackage{textcomp}
\usepackage{upgreek}
\usepackage{framed}
\usepackage{tabularx}
\usepackage{xcolor}
\usepackage{cases}
\usepackage{enumitem} 

\setlist[enumerate]{noitemsep,label=(\roman*),itemsep=1ex}

\usepackage{float} 	

\definecolor{linkblue}{cmyk}{1,.7,0,0}
\definecolor{linkred}{cmyk}{0,.92,.92,.27}
\definecolor{ownwhite}{cmyk}{1,1,1,1}

\usepackage{subfigure}
\usepackage{setspace} 
\usepackage{bbm} 	
\usepackage[normalem]{ulem}	
\usepackage{marvosym} 

\usepackage[all]{xy} 
\usepackage{tikz} 

\usepackage{mathtools} 

\usepackage{braket}  

\usepackage{bm}		

\usepackage{slashed}  

\usepackage{tensor}

\numberwithin{equation}{section}

\numberwithin{figure}{section}	
\numberwithin{table}{section}	

\usepackage[margin=0pt,font=small,labelformat=simple,labelfont=bf,labelsep=space, justification=justified, hypcap]{caption}
\usepackage[outercaption]{sidecap}

\usepackage{mdframed}
\usepackage[amsmath,thmmarks, framed]{ntheorem}
\usepackage{lipsum}
\usepackage{amssymb}

\usepackage{etoolbox}

\usepackage{lipsum}

\usepackage[hidelinks, linkcolor=linkblue, citecolor=linkblue, urlcolor=black, pdftitle={Quantization of the Proca field on curved spacetimes},
pdfauthor={Maximilian Schambach}]{hyperref}

\usepackage{ellipsis} 

\makeatletter
\let\nobreakitem\item
\let\@nobreakitem\@item
\patchcmd{\nobreakitem}{\@item}{\@nobreakitem}{}{}
\patchcmd{\nobreakitem}{\@item}{\@nobreakitem}{}{}
\patchcmd{\@nobreakitem}{\@itempenalty}{\@M}{}{}
\patchcmd{\@xthm}{\ignorespaces}{\nobreak\ignorespaces}{}{}
\patchcmd{\@ythm}{\ignorespaces}{\nobreak\ignorespaces}{}{}

\theoremframepreskip{0.3cm}
\theoremframepostskip{0.3cm}
\theorempreskip{0.3cm}
\theorempostskip{0.3cm}

\theoremstyle{plain}
\theoremheaderfont{\normalfont \bfseries}
\theorembodyfont{\itshape}

\theoreminframepreskip{0.4cm}
\theoreminframepostskip{0.4cm}
\newtheorem{theorem}{Theorem}[section]

\newtheorem{corollary}[theorem]{Corollary}
\newtheorem{lemma}[theorem]{Lemma}

\theoreminframepreskip{0.cm}
\theoreminframepostskip{0.cm}
\usepackage{etoolbox}
\theorembodyfont{\normalfont}
\theoremheaderfont{\normalfont \bfseries}
\newtheorem{definition}[theorem]{Definition}

\theorembodyfont{\normalfont}
\theoremheaderfont{\normalfont \bfseries}

\theorembodyfont{\normalfont}
\theoremheaderfont{\normalfont \bfseries}

\AtBeginEnvironment {conjecture}{}

\theorempreskip{0cm}
\theorempostskip{0.1cm}
\theoremstyle{nonumberplain}
\theoremheaderfont{\itshape}
\theorembodyfont{\normalfont}
\theoremsymbol{$\square$}
\theorempostwork{\bigskip}
\newtheorem{proof}{Proof:}[section]
\newtheorem{proof-idea}{Proof idea:}[section]

\definecolor{grey}{rgb}{0.925,0.925,0.925}

\DeclarePairedDelimiter\abs{\lvert}{\rvert}%
\DeclarePairedDelimiter\norm{\lVert}{\rVert}%

\makeatletter
\let\oldabs\abs
\def\abs{\@ifstar{\oldabs}{\oldabs*}}
\let\oldnorm\norm
\def\norm{\@ifstar{\oldnorm}{\oldnorm*}}
\makeatother

\newcommand{\e}{\ensuremath{\mathrm{e}}}

\newcommand{\PH}{\hspace{0.02em}\cdot\hspace{0.02em}}

\newcommand{\IC}{\mathbb{C}} 
\newcommand{\IR}{\mathbb{R}} 
\newcommand{\IN}{\mathbb{N}} 

\newcommand{\M}{\ensuremath{\mathcal{M}}}    
\newcommand{\N}{\ensuremath{\mathcal{N}}}    
\newcommand{\F}{\ensuremath{\mathcal{F}}}    
\newcommand{\A}{\ensuremath{\mathcal{A}}}    
\renewcommand{\AA}{\ensuremath{\mathscr{A}}}    

\newcommand{\BUOmega}{\ensuremath{{\mathcal{BU}\big(\Omega^1_0(\M)\big)}}}    
\newcommand{\BUmjdyn}{\ensuremath{\mathcal{BU}_{m,j}^\mathrm{dyn}}}    
\newcommand{\BUmzdyn}{\ensuremath{\mathcal{BU}_{m,0}^\mathrm{dyn}}}    
\newcommand{\BU}{\ensuremath{\mathcal{BU}}}    

\newcommand{\Dzs}{\ensuremath{\mathcal{D}_{0}(\Sigma)}}    

\newcommand{\Green}[2]{\ensuremath{\mathcal{G}_m(#1,#2)}}    
\newcommand{\IMG}[1]{\ensuremath{\mathrm{img}{}(#1)}}    

\newcommand{\IMJDYN}{\ensuremath{\mathcal{I}_{m,j}^\mathrm{\,dyn}}}    
\newcommand{\IMZDYN}{\ensuremath{\mathcal{I}_{m,0}^\mathrm{\,dyn}}}    
\newcommand{\IMJCCR}{\ensuremath{\mathcal{I}_{m,j}^\mathrm{\,CCR}}}    
\newcommand{\IMJZCCR}{\ensuremath{\mathcal{I}_{m,0}^\mathrm{\,CCR}}}    
\newcommand{\ICCR}{\ensuremath{{\mathcal{I}}_{\sim}^\mathrm{\,CCR}}}    
\newcommand{\ICCRS}{\ensuremath{{\mathcal{I}}_{\sim}^{\mathrm{\,CCR},\Sigma}}}    
\newcommand{\ICCRSP}{\ensuremath{{\mathcal{I}}_{\sim}^{\mathrm{\,CCR},\Sigma'}}}    
\newcommand{\IMJ}{\ensuremath{\mathcal{I}_{m,j}}}    
\newcommand{\IMZ}{\ensuremath{\mathcal{I}_{m,0}}}    

\newcommand{\KERN}[1]{\ensuremath{\mathrm{ker}{\left(#1\right)}}}   
\newcommand{\SPAN}[1]{\ensuremath{\mathrm{span}{\big\{#1\big\} }}}

\newcommand{\TsS}{\ensuremath{T^*\Sigma}}    

\DeclareDocumentCommand\Gm{ m g }{%
	{\ensuremath{\mathcal{G}_m %
		\IfNoValueF {#2} {(#1, #2)}%
	}%
}
}

\DeclareDocumentCommand\Gmz{ m g }{%
	{\ensuremath{\mathcal{G}_{m_0} %
			\IfNoValueF {#2} {(#1, #2)}%
		}%
	}
}

\DeclareDocumentCommand\Em{ m g }{%
	{\ensuremath{\mathcal{E}_m %
			\IfNoValueF {#2} {(#1, #2)}%
		}%
	}
}

\DeclareDocumentCommand\Ez{ m g }{%
	{\ensuremath{\mathcal{E}_{0} %
			\IfNoValueF {#2} {(#1, #2)}%
		}%
	}
}

\usepackage{xfrac}
\newcommand{\Quotient}[2]{\ensuremath{\sfrac{#1}{#2}}}    

\newcommand{\Quotientscale}[2]{\ensuremath{{\scalebox{1.2}{\Quotient{#1}{#2}}}}}    

\newcommand{\comp}{\mathbin{\mathchoice
  {\xcirc\scriptstyle}
  {\xcirc\scriptstyle}
  {\xcirc\scriptscriptstyle}
  {\xcirc\scriptscriptstyle}
}}
\newcommand{\xcirc}[1]{\vcenter{\hbox{$#1\circ$}}}

\newcommand{\rhoz}{\ensuremath{\rho_{(0)}}}    
\newcommand{\rhon}{\ensuremath{\rho_{(n)}}}    
\newcommand{\rhod}{\ensuremath{\rho_{(d)}}}    
\newcommand{\rhodelta}{\ensuremath{\rho_{(\delta)}}}    

\newcommand{\Az}{\ensuremath{{A_{(0)}}}}    
\newcommand{\An}{\ensuremath{{A_{(n)}}}}    
\newcommand{\Ad}{\ensuremath{{A_{(d)}}}}    
\newcommand{\Adelta}{\ensuremath{{A_{(\delta)}}}}    

\newcommand{\supp}[1]{\ensuremath{\mathrm{supp}\left(#1\right)}}    

\newcommand{\formspace}{\,}

\renewcommand{\i}{\ensuremath{\mathrm{i}}}

\newcommand{\SpacCurr}{\ensuremath{\mathsf{SpacCurr}}}  
\newcommand{\Alg}{\ensuremath{\mathsf{Alg}}}  

\newcommand\restr[2]{{
  \left.\kern-\nulldelimiterspace 
  #1 
  \vphantom{\big|} 
  \right|_{#2} 
  }}

\newbox\usefulbox

\makeatletter
\def\getslant #1{\strip@pt\fontdimen1 #1}

\def\skoverline #1{\mathchoice
 {{\setbox\usefulbox=\hbox{$\m@th\displaystyle #1$}%
    \dimen@ \getslant\the\textfont\symletters \ht\usefulbox
    \divide\dimen@ \tw@ 
    \kern\dimen@ 
 \hspace{2pt}   \overline{\mkern -3mu \kern-\dimen@ \box\usefulbox\kern\dimen@ \mkern -1mu}\kern-\dimen@ }}
 {{\setbox\usefulbox=\hbox{$\m@th\textstyle #1$}%
    \dimen@ \getslant\the\textfont\symletters \ht\usefulbox
    \divide\dimen@ \tw@ 
    \kern\dimen@ 
 \hspace{2pt}  \overline{\mkern -3mu \kern-\dimen@ \box\usefulbox\kern\dimen@ \mkern -1mu}\kern-\dimen@ }}
 {{\setbox\usefulbox=\hbox{$\m@th\scriptstyle #1$}%
    \dimen@ \getslant\the\scriptfont\symletters \ht\usefulbox
    \divide\dimen@ \tw@ 
    \kern\dimen@ 
 \hspace{2pt} \overline{\mkern -3mu \kern-\dimen@ \box\usefulbox\kern\dimen@ \mkern -1mu}\kern-\dimen@ }}
 {{\setbox\usefulbox=\hbox{$\m@th\scriptscriptstyle #1$}%
    \dimen@ \getslant\the\scriptscriptfont\symletters \ht\usefulbox
    \divide\dimen@ \tw@ 
    \kern\dimen@ 
 \hspace{2pt}  \overline{\mkern -3mu \kern-\dimen@ \box\usefulbox\kern\dimen@ \mkern -1mu}\kern-\dimen@ }}%
 {}}
\makeatother

\usepackage[backend=biber, style=numeric, citestyle=numeric, language=english, sortlocale=de_DE, natbib=true, url=false, doi=false, eprint=true, isbn=false, maxbibnames=99, sorting=nyt, giveninits=true]{biblatex}  
\renewbibmacro{in:}{}
\DeclareFieldFormat[article, book, misc, incollection]{title}{\textit{#1}}			
\DeclareFieldFormat[article, incollection]{journaltitle}{#1}
\DeclareFieldFormat[article, incollection]{pages}{#1}
\DeclareFieldFormat[article, incollection]{volume}{\textbf{#1}}
\DeclareFieldFormat[incollection]{booktitle}{In: \textit{#1}}

\usepackage{authblk}

\begin{document}


\author[1]{Maximilian Schambach%
	\thanks{email: \texttt{schambach@kit.edu}}}
\author[2]{Ko Sanders%
	\thanks{email: \texttt{jacobus.sanders@dcu.ie}}}

\affil[1]{\normalsize Institute of Industrial Information Technology\\ Karlsruhe Institute of Technology}
\affil[2]{\normalsize School of Mathematical Sciences\\ Dublin City University}

\date{\large  \today}

\maketitle
\thispagestyle{empty}

\pagenumbering{arabic}

\hypersetup{linkcolor=linkblue}


%
%
\begin{abstract}
We investigate the classical and quantum Proca field (a massive vector potential) of mass $m>0$ in arbitrary globally hyperbolic spacetimes and in the presence of external sources. We motivate a notion of continuity in the mass for families of observables $\left\{O_m\right\}_{m>0}$ and we investigate the massless limit $m\to0$. Our limiting procedure is local and covariant and it does not require a choice of reference state. We find that the limit exists only on a subset of observables, which automatically implements a gauge equivalence on the massless vector potential. For topologically non-trivial spacetimes, one may consider several inequivalent choices of gauge equivalence and our procedure selects the one which is expected from considerations involving the Aharonov-Bohm effect and Gauss' law.\par
We note that the limiting theory does not automatically reproduce Maxwell's equation, but it can be imposed consistently when the external current is conserved. To recover the correct Maxwell dynamics from the limiting procedure would require an additional control on limits of states. We illustrate this only in the classical case, where the dynamics is recovered when the Lorenz constraint remains well behaved in the limit.
\end{abstract}

\section{Introduction}\label{chpt:introduction}
Massive vector potentials satisfying Proca's equation are the most straightforward massive generalization of the massless vector potential of electromagnetism. They may be used for an effective description of vector particles in the standard model, such as W- and Z-bosons (who really acquire their mass through the Higgs mechanism), or as a modification of the massless photon. In the latter scenario, the Proca field provides a theoretical framework to study upper bounds on the photon mass. It is important to note, however, that the Proca field does not have a gauge symmetry, unlike the massless vector potential of electromagnetism\footnote{An alternative approach to massive electrodynamics due to Stueckelberg preserves the gauge invariance by introducing an extra scalar field, cf. \cite{BelokogneFolacci2016}.}. \par
In this paper we will make a theoretical investigation of the massless limit of the Proca field in curved spacetimes, with special attention to the emergence of the gauge symmetry. In Minkowski space this massless limit is textbook material (cf. \cite{ItzyksonZuber1980}), but the corresponding problem in curved spacetimes poses some additional interesting challenges, which we now discuss.

Firstly, to define the quantum Proca field we cannot avail ourselves of a vacuum state or a preferred Hilbert space representation for the quantum theory. However, it is well understood how to circumvent this problem using an algebraic approach. On a given spacetime we can then describe the Proca field of mass $m>0$ with an external current $j$ by an abstract $^*$-algebra $\AA_{m,j}$. For $j=0$ such a construction has already been given  by Furlani \cite{Furlani1999}, imposing some topological restrictions, and later by Dappiaggi \cite{Dappiaggi2011}. The methods needed to include non-trivial currents $j$ are also well known in principle, see e.g. \cite{SandersDappiaggiHack2014} or \cite{FewsterSchenkel2015}. In this paper we will not pursue the investigation of states and Hilbert space representations, which forms the next step in the description of the quantum theory.\par
Secondly, to define a notion of continuity in the mass, we will need to compare the algebras $\AA_{m,j}$ at different values of $m$. Once again we cannot resort to preferred vacuum states or Hilbert space representations. Instead we will propose a notion of continuity in the mass for families of observables $\left\{O_m\right\}_{m>0}$, which is formulated entirely at the algebraic level. This continuity makes use of the fact that for all $m>0$ the algebras $\AA_{m,j}$ are isomorphic to an algebra of initial data on a Cauchy surface, which is independent of $m$. We prove that our notion of continuity is independent of the choice of Cauchy surface before we define the massless limit of the Proca field.\par
Thirdly, the gauge freedom of free electromagnetism admits at least three generalisations from Minkowski space to spacetimes with non-trivial topologies. One may use e.g. the field strength tensor $F$, or equivalence classes of one-forms $A$, where the pure gauge solutions are either the closed or the exact one-forms. One of us has previously argued that the latter choice is the preferred one in a generally covariant setting, because it allows the correct description of the Aharonov-Bohm effect and Gauss' law \cite{SandersDappiaggiHack2014}. We will show that this choice of gauge equivalence also arises naturally from the limiting procedure, thereby providing an additional justification for it.\par
\smallskip
In Section \ref{sec:review} below we will review the classical and the quantum Proca field in an arbitrary globally hyperbolic spacetime with a fixed mass $m>0$ and external current $j$. In Section \ref{sec:limit} we will then formulate the continuity in the mass and define the zero mass limit. We will show that this limit exists only for a certain sub-algebra of observables and, by choosing this algebra as large as possible, we automatically arrive at the gauge equivalence given by exact forms, as preferred by \cite{SandersDappiaggiHack2014}. In this section we also comment on the fact that the zero mass limit yields a theory that does not automatically include Maxwell's equations. We believe that this is due to the fact that we did not include the behaviour of states in the zero mass limit, and we illustrate this with an argument concerning the classical Proca field. Although it may be possible to include classes of states (e.g. Hadamard states \cite{FewsterPfenning2003}) and to study their behaviour during a limiting process, we will not pursue this in the present investigation. Section \ref{chpt::conclusion} contains our conclusions and a brief outlook.\par
\smallskip
We will use the remainder of this section to introduce some conventions and notations that will be used throughout the paper. We let $(\M,g)$ denote a spacetime, consisting of a smooth, four dimensional manifold $\M$, assumed to be Hausdorff, connected, oriented and para-compact,  and a Lorentzian metric $g$, whose signature is chosen to be $(-,+,+,+)$. We assume that $(\M,g)$ is globally hyperbolic and time-oriented. A generic smooth, space-like Cauchy surface is denoted by $\Sigma$, with an induced Riemannian metric $h$. The Levi-Civita connection on $(\M,g)$ will be denoted by $\nabla$ and the one on $\Sigma$ by $\nabla_{(\Sigma)}$. For further standard notations regarding spacetimes (e.g. causal relations and tensor calculus) we refer to \cite{Wald1984}.\par
The space of smooth differential forms on $\M$ of degree $p$ will be denoted by $\Omega^p(\M)$, and the subspace of compactly supported forms by $\Omega^p_0(\M)$. The space of all differential forms is an algebra under the exterior product $\wedge$. Using the metric we can define a Hodge $*$-operation such that $A\wedge *B=\frac{1}{p!}A^{\mu_1\ldots\mu_p}B_{\mu_1\ldots\mu_p}d\mathrm{vol}_g$, where $d\mathrm{vol}_g$ is the natural volume form determined by the metric. We may define a pairing on the space of $p$-forms by
\begin{align}
\langle A,B\rangle_{\M} \coloneqq \int_{\M} A\wedge {*B}
\end{align}
when the support of $A\wedge {*B}$ is compact. The pairing is symmetric, $\langle A,B\rangle_{\M}=\langle B,A\rangle_{\M}$, and it defines an inner product on the spaces $\Omega^p_0(\M)$.

The co-derivative $\delta$ is defined in terms of the exterior derivative $d$ by $\delta \coloneqq (-1)^{s+1+n(p-1)}{*d*}$ when acting on $p$-forms, where $n$ is the dimension of the manifold ($n=4$ on $\M$ and $n=3$ on the Cauchy surface $\Sigma$) and $s$ is the number of negative eigenvalues of the metric ($s=1$ on $\M$ and $s=0$ on $\Sigma$). One may show that $\delta$ and $d$ are each other's (formal) adjoints under the pairing $\langle \PH,\PH\rangle_{\M}$. The Laplace-Beltrami operator on $p$-forms is defined by $\square=d\delta+\delta d$, which is a normally hyperbolic operator. A form $A$ is called closed when $dA=0$ and exact when $A=dB$ for some differential form $B$. It will be convenient to denote the space of closed $p$-forms on $\M$ by $\Omega^p_d(\M)$ and the compactly supported closed $p$-forms by $\Omega^p_{0,d}(\M)$. Similarly, $A$ is called co-closed when $\delta A=0$ and co-exact when $A={\delta B}$ for some differential form $B$. Once again it will be convenient to denote the space of co-closed forms on $\M$ by $\Omega^p_{\delta}(\M)$ and the compactly supported co-closed $p$-forms by $\Omega^p_{0,\delta}(\M)$. For more details on differential forms we refer the reader to \cite{BottTu1982}.

\section{The Proca field in curved spacetimes}\label{sec:review}

\subsection{The classical Proca field in curved spacetimes}
Let $A, j \in \Omega ^1 (\M)$ be smooth one-forms on $\M$ and $m >0$ a positive constant. We will call $A$ the \emph{Proca field}, $m$ its \emph{mass} and $j$ an \emph{external current}. The \emph{Proca equation} reads:
\begin{equation}
\left( \delta d + m^2 \right) A = j \formspace.\label{eqn:Proca}
\end{equation}
Accordingly, the \emph{Proca operator} is defined as $(\delta d + m^2)$.
It is well known that the Proca operator is Green-hyperbolic but \emph{not} normally hyperbolic \cite{Baer2015}. However, we can decompose Proca's equation into a wave equation and a Lorenz constraint:
\begin{align}
\left(\square +m^2 \right) A &= j + m^{-2} \, d \delta j \label{eqn:classical_wave_eqation} \formspace,\\
\delta A &= m^{-2} \delta j \label{eqn:classical_constraint} \formspace,
\end{align}
which together are equivalent to the Proca equation (\ref{eqn:Proca}) when $m>0$. Indeed, applying $\delta$ to (\ref{eqn:Proca}) yields (\ref{eqn:classical_constraint}), and in the presence of this equality, (\ref{eqn:Proca}) and (\ref{eqn:classical_wave_eqation}) are equivalent.
Following Dimock \cite{Dimock1992}, Furlani \cite{Furlani1999} and Pfenning \cite{Pfenning2009} we parametrise the initial data of differential forms with the following operators:
\begin{definition}\label{def:cauchy_mapping_operators}
	Let $i : \Sigma \hookrightarrow \M$ be the inclusion  of the Cauchy surface $\Sigma$ with pullback~$i^*$.
	The operators $\rhoz, \rhod : \Omega^p(\M) \to \Omega^p(\Sigma)$ and $\rhon, \rhodelta : \Omega^p(\M) \to \Omega^{p-1}(\Sigma)$ are defined as:
	\begin{equation}
		\rhoz 		= i^* 	\formspace,\quad
		\rhod 		= -{*_{(\Sigma)} i^*} {* d} 		\formspace,\quad
		\rhodelta  = i^*\delta 						\quad\text{and}\quad
		\rhon 		= -{*_{(\Sigma)} i^* *}  		\formspace.
	\end{equation}
	Let $A \in \Omega^1(\M)$. The differential forms $\Az, \Ad \in \Omega^1(\Sigma)$ and $\An, \Adelta \in \Omega^0(\Sigma)$ are defined as:
	\begin{equation}
	\Az = \rhoz A \formspace, \quad
	\Ad = \rhod A \formspace, \quad
	\An = \rhon A \quad\textrm{and}\quad
	\Adelta = \rhodelta A \formspace.
	\end{equation}
\end{definition}
Specifying these differential forms is equivalent to specifying the initial data $A_\mu$ and $n^\alpha \nabla_\alpha A_\mu$ on the Cauchy surface $\Sigma$ with future pointing unit normal vector field $n$ \cite{Furlani1999}. \par

The wave operator $(\square + m^2)$ on $p$-forms has unique advanced $(-)$ and retarded $(+)$ fundamental solutions $E_m^\pm : \Omega^p_0(\M) \to \Omega^p(\M)$ with $\supp{E^\pm_m F} \subset J^\pm (\supp{F})$ \cite{BaerGinouxPfaeffle2007}. It is straightforward to show that the fundamental solutions intertwine their action with the interior and exterior derivative, i.\,e. it holds $E_m^\pm d = d E_m^\pm$ and $E_m^\pm \delta = \delta E_m^\pm$. The advanced minus retarded fundamental solution is denoted by $E_m = E_m^- - E_m^+$.\par 
With the notion of the fundamental solutions we can state a solution to the wave equation (\ref{eqn:classical_wave_eqation}) in form of the following
\begin{theorem}[Solution of the wave equation]
	Let $\Az,\Ad \in \Omega^1(\Sigma)$ and $\An,\Adelta\in\Omega^0(\Sigma)$ specify initial data on the Cauchy surface $\Sigma$. Let $F \in \Omega^1_0(\M)$ be a test one-form and $\kappa \in \Omega^1(\M)$ an external source. Then, for any $m\ge0$,
	\begin{align}
	\langle A, F \rangle_\M = \sum\limits_\pm \langle E_m^\mp F , \kappa \rangle_{J^{\pm}(\Sigma)}
	&- \langle \Az , \rhod E_m F \rangle_\Sigma
	- \langle \Adelta , \rhon E_m F \rangle_\Sigma \notag \\
	&+ \langle \An , \rhodelta E_m F \rangle_\Sigma
	+ \langle \Ad , \rhoz E_m F \rangle_\Sigma
	\end{align}
	specifies the unique smooth solution $A\in \Omega^1(\M)$ of the wave equation $(\square + m^2)A= \kappa$ with the given initial data. Furthermore, the solution depends continuously on the initial data. \label{thm:solution_wave_equation}
\end{theorem}
The proof is a straightforward generalization of the source free case \cite{Furlani1999}, see
e.\,g.~Theorem 2.3 and Lemma 2.4 of \cite{SandersDappiaggiHack2014}. \par
Now that we have solved the wave equation (\ref{eqn:classical_wave_eqation}), we turn to the Lorenz constraint (\ref{eqn:classical_constraint}). Assume that $A\in \Omega^1(\M)$ solves the wave equation, $(\square +m^2)A =  j + m^{-2} d \delta j$. We observe
\begin{align}
(\square +m^2) \delta A
&= \delta (\square + m^2) A
= \delta \left( j + m^{-2} d \delta j \right) \notag\\
&= (\square + m^2 )  m^{-2} \delta j  \label{eqn:waveeqn_constraint}
\formspace.
\end{align}
The solution $A$ to the wave equation therefore yields a Klein-Gordon equation for $\delta A - m^{-2} \delta j$. This ensures that the Lorenz constraint (\ref{eqn:classical_constraint}) propagates and, to impose the constraint (and hence obtain a solution to Proca's equation), it suffices to require that the initial data of $\delta A - m^{-2} \delta j$ vanish on the Cauchy surface $\Sigma$ \cite[Cor. 3.2.4]{BaerGinouxPfaeffle2007}. We will re-express this requirement in terms of constraints on initial data of $A$, making use of the following two lemmas:
\begin{lemma}\label{lem:klein_gordon_cauchy_data}
	Let $\Sigma$ be a Cauchy surface with unit normal vector field $n$. For any smooth zero-form $f\in \Omega^0(\M)$ it holds that
	\begin{equation}
		\rhon f = 0 \formspace, \quad
		\rhodelta f = 0 \formspace, \quad
		\rhoz f = \restr{f}{\Sigma} \formspace, \quad\textrm{and}\quad
		\rhod f = \restr{(df)(n)}{\Sigma} = \restr{\left(n^\alpha \nabla_\alpha f\right)}{\Sigma} \formspace.
	\end{equation}
	Therefore, with respect to the Klein Gordon equation, $\rhoz f$ and $\rhod f$ specify initial data on $\Sigma$.
\end{lemma}
\begin{proof}
The proof of these identities is straightforward (cf.~\cite[Lemma 3.8]{Schambach2016}). For example, $\rhod f=\rhon df=
\restr{n^\alpha (df)_\alpha}{\Sigma}$ by \cite[Appendix A]{Furlani1999}, and $(df)_{\alpha}=\nabla_\alpha f$.
\end{proof}
\begin{lemma}[Gaussian Coordinates] \label{lem:normal_vectors}
	Let $\Sigma$ be a Cauchy surface of $\M$ with future pointing unit normal vector field $n$. We can extend $n$ to a neighbourhood of $\Sigma$ such that
	\begin{equation}
	n^\alpha \nabla_\alpha n^\beta = 0 \formspace, \quad
	dn = 2 \nabla_{[\mu} n_{\nu]} = 0 \formspace. \label{eqn:gauss_coordinates}
	\end{equation}
\end{lemma}
\begin{proof}
	An introduction to Gaussian (normal) coordinates is for example given in \cite[pp. 42,43]{Wald1984} or \cite[pp. 445,446]{Carroll2004} where the first equation of (\ref{eqn:gauss_coordinates}) is shown to hold by construction. The second Equation of (\ref{eqn:gauss_coordinates}) can be derived by using Frobenius' theorem (see for example \cite[Theorem B.3.1 and B.3.2]{Wald1984}) as explained in \cite[Section 2.3.2 Equation (5)]{SandersDappiaggiHack2014}.
\end{proof}
With this normal vector field we can write the metric $g$ of the spacetime $\M$ in a neighbourhood of the Cauchy surface as $g\indices{_{\mu\nu}} = - n_\mu n_\nu + h\indices{_{\mu\nu}}$ \cite[Equation 10.2.10]{Wald1984}, where $h$ extends the induced metric on $\Sigma$.

To state the main result of this section we also introduce fundamental solutions for the Proca operator $(\delta d + m^2)$. The Proca operator, being Green hyperbolic, has unique advanced $(-)$ and retarded $(+)$ fundamental solutions $G_m^\pm : \Omega^p_0(\M) \to \Omega^p(\M)$ which are given in terms of the fundamental solutions of the wave operator by
\begin{equation}
G_m^\pm = (m^{-2}{d\delta} +1) E_m^\pm \formspace,
\end{equation}
cf. \cite[Example 2.17]{BaerGinoux2012}. Analogously we define $G_m = (m^{-2}{d\delta} +1) E_m$. We then have
\begin{theorem}[Solution of Proca's equation]\label{thm:solution_proca_unconstrained}
	Let $\Az,\Ad \in \Omega^1(\Sigma)$ on $\Sigma$, $F \in \Omega^1_0(\M)$ a test one-form, $j \in \Omega^1(\M)$ an external source and $m >0$ a mass. Then,
	\begin{equation}
	\langle A, F \rangle_\M = \sum\limits_\pm \langle j , G_m^\mp F   \rangle_{J^{\pm}(\Sigma)}
	- \langle \Az , \rhod G_m F \rangle_\Sigma
	+ \langle \Ad , \rhoz G_m F \rangle_\Sigma \label{eqn:solution_proca_unconstrained}
	\end{equation}
	specifies the unique smooth solution of Proca's equation $\left( \delta d + m^2 \right) A = j$ with the given $\Az$ and $\Ad$. Furthermore, the solution depends continuously on these initial data, and we have
		\begin{equation}\label{eqn:proca_constraints}
		\Adelta = m^{-2} \rhodelta j \quad \text{and} \quad
		m^2 \, \An = \rhon j + \delta_{(\Sigma)}\Ad \formspace.
		\end{equation}
\end{theorem}
\begin{proof}
We use the equivalence of the Proca equation (\ref{eqn:Proca}) with the wave equation (\ref{eqn:classical_wave_eqation}) and the vanishing of the initial data of $\delta A-m^{-2}\delta j$. We let $A$ be a solution of the wave equation $(\square + m^2) A = \kappa$ with $\kappa = j + m^{-2} \, d \delta j$. We first show that the specified constraints (\ref{eqn:proca_constraints}) on the initial data are equivalent to the vanishing of the initial data of $\delta A - m^{-2} \delta j$. For this we use Lemma \ref{lem:klein_gordon_cauchy_data} and \ref{lem:normal_vectors}.
The vanishing of the initial value yields, using the linearity of the pullback and Definition \ref{def:cauchy_mapping_operators}:
	\begin{equation}
	0 = \rhoz \left( \delta A - m^{-2} \delta j\right)
	= \rhodelta A - m^{-2} \rhodelta j \formspace.
	\end{equation}
We will calculate the vanishing of the normal derivative in Gaussian normal coordinates and in the end turn back to a coordinate independent notation:
	\begin{equation}
	0 = \rhod \left( \delta A - m^{-2} \delta j\right)
	= \restr{\Big(n^\alpha \nabla_\alpha  \delta A \Big)}{\Sigma} - m^{-2} \rhod  \delta j  \formspace. \label{eqn:cauchy_normal_derivative_tmp}
	\end{equation}
	We will take a separate look at the first summand:
	\begin{align}
	n^\alpha \nabla_\alpha  \delta A
	&= n^\alpha \left( d\delta A \right)_{\alpha}
	= n^\alpha \square A_\alpha - n^\beta \left( \delta  d A \right)_\beta \notag\\
	&= n^\alpha \kappa_\alpha -{m^2} \, n^\mu A_\mu + 2 n^\beta \nabla^\nu \nabla_{[\nu} A_{\beta]} \notag\\
	&= n^\alpha \kappa_\alpha - {m^2} \, n^\mu A_\mu +  2 \nabla^\nu \left( n^\beta  \nabla_{[\nu} A_{\beta]} \right) \formspace ,
	\end{align}
	where we have used that $\nabla^\nu n^\beta$ is symmetric by Lemma \ref{lem:normal_vectors}.
	Writing $g\indices{_{\mu\nu}} = -n_\mu n_\nu + h\indices{_{\mu\nu}}$ and using Lemma \ref{lem:normal_vectors} we find:
	\begin{align}
	\restr{g\indices{^{\sigma\nu}} \nabla_\sigma \left( n^\beta  \nabla_{[\nu} A_{\beta]} \right)}{\Sigma}
	&= \restr{\left(- n^\sigma n^\nu +  h\indices{^{\sigma\nu}} \right) \nabla_\sigma \left( n^\beta  \nabla_{[\nu} A_{\beta]} \right) }{\Sigma} \notag \\
	&= 0 + \restr{\nabla_{(\Sigma)}^{\nu} \left( n^\beta  \nabla_{[\nu} A_{\beta]} \right)}{\Sigma} \formspace.
	\end{align}
	Here we have made use of the identification $ h\indices{^{\sigma\nu}} \nabla_\sigma B_\mu= \nabla_{(\Sigma)}^\nu  B_\mu$ for any one-form $B$ tangential to $\Sigma$ \cite[Lemma 10.2.1]{Wald1984}.
	We identify $2 n^\beta  \nabla_{[\nu} A_{\beta]} = -2 n^\beta  \nabla_{[\beta} A_{\nu]} =- A_{(d)\nu}$ and use $\delta_{(\Sigma)} B =  -\nabla_{(\Sigma)} ^\alpha B_\alpha$ to obtain
	\begin{equation}
	\restr{n^\alpha \nabla_\alpha  \delta A}{\Sigma}
	 =\rhon \kappa - m^2 \An  + \delta_{(\Sigma)} \Ad \formspace.
	\end{equation}
	Inserting this into Equation (\ref{eqn:cauchy_normal_derivative_tmp}) and using the definition of the source term $\kappa = j + m^{-2} {d \delta j}$, we find from $\rhod = \rhon d$ that
	\begin{equation}
	m^2 \An
	 = \rhon \kappa - m^{-2} \rhod \delta j + \delta_{(\Sigma)} \Ad
	=  \rhon j + \delta_{(\Sigma)} \Ad \formspace.
	\end{equation}
This proves that (\ref{eqn:proca_constraints}) are the required constraints.

We now substitute the constraints (\ref{eqn:proca_constraints}) in the formula of Theorem \ref{thm:solution_wave_equation} and show that we recover Equation (\ref{eqn:solution_proca_unconstrained}). We find
	\begin{align} \label{eqn:tmp_proca_solution_constraint_inserted}
	\langle A, F \rangle_\M =
	& \sum\limits_\pm \langle j + m^{-2} d\delta j, E_m^\mp F   \rangle_{J^{\pm}(\Sigma)}
	- \langle \Az , \rhod E_m F \rangle_\Sigma
	- m^{-2} \langle  \rhodelta j, \rhon E_m F \rangle_\Sigma \notag \\
	& + m^{-2} \langle \delta_{(\Sigma)} \Ad , \rhodelta E_m F \rangle_\Sigma
	+ m^{-2} \langle \rhon j, \rhodelta E_m F \rangle_\Sigma
	+ \langle \Ad, \rhoz E_m F \rangle_\Sigma \formspace.
	\end{align}
	Now, for clarity's sake, we take a look at the appearing terms separately. To get rid of the divergence of $\Ad$, we use the formal adjointness of $\delta$ and $d$ and the commutativity of $d$ with the pullback $i^*$:
	\begin{align}
	m^{-2}\langle \delta_{(\Sigma)} \Ad, \rhodelta E_m F \rangle_\Sigma
	=& m^{-2}\langle \Ad, d_{(\Sigma)} i^* \delta  E_m F \rangle_\Sigma \notag\\
	=& m^{-2}\langle \Ad,  i^* d \delta  E_m F \rangle_\Sigma \notag \\
	=& m^{-2}\langle \Ad, \rhoz d \delta  E_m F \rangle_\Sigma \formspace,
	\end{align}
	which, together with $\langle \Ad, \rhoz E_m F \rangle_\Sigma$,  combines to $\langle \Ad , \rhoz G_m F \rangle_\Sigma$. \par
	Next, we have a look at a part of the sum term and use Stoke's theorem (we get a sign $\mp$ due to the orientation of $\Sigma$ with respect to $J^{\pm}(\Sigma)$) for a partial integration, at the cost of some boundary terms:
    \begin{align}\label{eqn:partialintegration}
    \sum\limits_\pm &\langle d \delta j , E_m^\mp F \rangle_{J^{\pm}(\Sigma)} -
    \langle j , d \delta E_m^\mp F \rangle_{J^{\pm}(\Sigma)}\nonumber\\
    &= \sum\limits_\pm \int_{J^{\pm}(\Sigma)} d \delta j \wedge {* E_m^\mp} F -
    d \delta E_m^\mp F \wedge {* j}\nonumber\\
    &= \sum\limits_\pm \int_{J^{\pm}(\Sigma)} d \left(\delta j \wedge {* E_m^\mp} F -
    \delta E_m^\mp F \wedge {* j}\right)
    + \delta j \wedge {* \delta E_m^\mp} F - \delta E_m^\mp F \wedge {* \delta j}\nonumber\\
    &= \sum\limits_\pm \mp \int_\Sigma i^*\left(\delta j \wedge {* E_m^\mp } F - \delta E_m^\mp F \wedge {* j}\right)\nonumber\\
    &= -\int_\Sigma i^*\left(\delta j \wedge {* E_m F} - \delta E_m F \wedge {* j}\right).\nonumber\\
    &= -\langle i^*\delta j, *_{(\Sigma)} {i^* *} E_m F\rangle_{\Sigma} + \langle i^*\delta E_m F, *_{(\Sigma)} {i^* *} j\rangle_{\Sigma}\nonumber\\
    &=\langle \rhodelta j, \rhon E_m F\rangle_{\Sigma} - \langle \rhodelta E_m F, \rhon j\rangle_{\Sigma} \formspace.
    \end{align}
Multiplying this equality by $m^{-2}$ and rearranging, we see that the first, third and fifth terms of (\ref{eqn:tmp_proca_solution_constraint_inserted}) combine to the first term of (\ref{eqn:solution_proca_unconstrained}). Finally, we note that in the second term of (\ref{eqn:tmp_proca_solution_constraint_inserted}), $d^2 =0$ implies
	\begin{equation}
	\rhod G_m = - {*_{(\Sigma)} {i^* *}} d\left( m^{-2}{d \delta} +1 \right) E_m = - {*_{(\Sigma)}{i^* *}} d E_m = \rhod E_m 
	\end{equation}
which completes the proof.
\end{proof}
\subsection{The quantum Proca field in curved spacetimes}\label{sec:quantumProca}
The procedure to quantize the Proca field in a generally covariant way in the framework of Brunetti, Fredenhagen and Verch \cite{BrunettiFredenhagenVerch2003} is well understood, see e.\,g. \cite{Dappiaggi2011} for the source free case. The modifications needed to account for external currents can be made analogously to \cite{SandersDappiaggiHack2014} (see also \cite{FewsterSchenkel2015}). Throughout this section, the mass $m>0$ is assumed to be fixed. For simplicity we will mostly consider a single fixed spacetime $(\M,g)$ and source $j\in\Omega^1(\M)$. \par
The quantum Proca field is then described by the following algebra:
\begin{definition}\label{def:algebra-A(M)}
    The unital $^*$-algebra $\AA_{m,j}$ is obtained from the free algebra, generated by $\mathbbm{1}$ and the objects $\A_{m,j}(F)$, $F\in\Omega_0^1(\M)$, by factoring out the relations
	\begin{subequations}  \label{def:ideal_generators}
		\begin{align}
		\text{(i)}\; &\A_{m,j}(c F + c' F') = c\,\A_{m,j}(F) + c' \A_{m,j}(F') 														&\textrm{linearity,} \\
		\text{(ii)}\; &\A_{m,j}(F)^* = \A_{m,j}(\skoverline{F}\,) 																															&\textrm{hermitian field,} \\
		\text{(iii)}\; &\A_{m,j}\big( (\delta d + m^2) F \big) = \langle j, F \rangle_\M \cdot \mathbbm{1} 	&\textrm{equation of motion,} \\
		\text{(iv)}\; &[\A_{m,j}(F) , \A_{m,j}(F') ] = \i \Green{F}{F'} \cdot \mathbbm{1}															&\textrm{commutation relations},
		\end{align}
	\end{subequations}
for all $c, c' \in \IC$ and $F,F' \in \Omega^1_0(\M)$, where we write $\Green{F}{F'} = \langle F, G_mF' \rangle_\M$.
\end{definition}
For our later investigation of the zero mass limit it will be useful to describe the algebra $\AA_{m,j}$ and its topology in more detail in the next few sections.

\subsubsection{The Borchers-Uhlmann algebra}\label{sec:BU-algebra}

The algebra $\AA_{m,j}$ is obtained as a quotient of the \emph{Borchers-Uhlmann algebra} (BU-algebra), which is defined\footnote{Here, $\otimes$ denotes the algebraic tensor product, without taking any topological completion.} as the tensor algebra of the vector space $\Omega^1_0(\M)$,
\begin{equation}
\BUOmega \coloneqq \bigoplus_{n=0}^{\infty}\big(\Omega^1_0(\M)\big)^{\otimes n} \formspace.
\end{equation}
Elements $f \in \BUOmega$ are tuples $f= \big(f^{(0)}, f^{(1)}, f^{(2)}, \dots\big)$, where the components $f^{(0)} \in \IC$ and for $f^{(n)} \in \big(\Omega^1_0(\M)\big)^{\otimes n}$ for $n > 0$ such that only finitely many $f^{(n)}$'s are non-vanishing. We will call the component $f^{(n)}$ the \emph{degree-n-part} of $f$.

Addition and scalar multiplication in $\BUOmega$ are defined component-wise, and we can define a (tensor) product and $*$-operation by defining their degree-$n$-parts as
\begin{align}
(f\cdot g)^{(n)}(p_1,p_2, \dots , p_n) &=  \sum_{i+j = n} f^{(i)} (p_1, p_2, \dots , p_i) g^{(j)} (p_{i+1} , \dots , p_n) \formspace,\\
(f^*)^{(n)}(p_1, \dots , p_n) &=  \overline{f^{(n)}(p_n, p_{n-1}, \dots , p_1)}\formspace
\end{align}
for all elements $f,g$ and $p_i\in\M$. This makes $\BUOmega$ a *-algebra with unit element $\mathbbm{1}_{\BU(\Omega^1_0(\M))} = (1, 0, 0, \dots)$. The BU-algebra can be endowed with a locally convex topology \cite{SahlmannVerch2000}, obtained from the locally convex topology of\footnote{For a construction, see \cite[Chapter 17.1 to 17.3]{Dieudonne1972}.} $\Omega^1_0(\M)$. More precisely, we can view it as a dense sub-algebra of the \emph{complete BU-algebra}
\begin{equation}
\overline{\BUOmega} \coloneqq \bigoplus_{n=0}^{\infty}\Gamma_0((T^*\M)^{\boxtimes n}) \formspace,
\end{equation}
where $(T^*M)^{\boxtimes n}$ denotes the $n$-fold outer product bundle over $\M ^n$ (cf. \cite[Chapter 3.3]{SahlmannVerch2000}).\\ We note that the multiplication in $\overline{\BUOmega}$ is a jointly continuous bilinear map and hence so is the product in $\BUOmega$.

We want to identify smeared quantum fields $\A_{m,j}(F)$ with elements $(0, F, 0, 0, \dots)$, but the BU-algebra $\BUOmega$ incorporates neither any dynamics, nor the desired quantum commutation relations. It will be convenient to implement the Proca equation (in a distributional sense) and the canonical commutation relations (CCR) in a two step procedure. \par
First we divide out the two-sided ideal $\IMJDYN$ in $\BUOmega$ that is generated by elements
\begin{equation}
\big(-\langle j, F \rangle_\M, (\delta d + m^2)F,0,0,\dots\big) \in \BUOmega\formspace,
\end{equation}
for $F \in \Omega^1_0(\M)$, to implement the dynamics. That means, by definition, that every $f\in\IMJDYN$ can be written as a finite sum
\begin{equation}\label{eqn:ideal}
f = \sum_i g_i \cdot \left(-\langle j, F_i \rangle_\M, (\delta d + m^2)F_i,0,0,\dots\right) \cdot h_i \formspace,
\end{equation}
for some $F_i \in \Omega^1_0(\M)$ and $g_i, h_i \in \BUOmega$. We define
\begin{equation}
\BUmjdyn \coloneqq {\Quotientscale{\BUOmega}{\IMJDYN}} \formspace.
\end{equation}
Elements $f \in \BUmjdyn$ are then equivalence classes $ f = \left[ g \right]_{m,j}^\text{dyn}$ where $g \in \BUOmega$. \par

Now, in the second step, we incorporate the CCR by dividing out the two-sided ideal $\IMJCCR$ that is generated by elements
\begin{equation}
\big[ \big(-\i \Gm{F}{F'}, 0 , F \otimes F' - F' \otimes F, 0 , 0 , \dots\big) \big]_{m,j}^\text{dyn} \in \BUmjdyn
\end{equation}
to obtain the final field algebra
\begin{equation}
\AA_{m,j} = {\Quotientscale{\BUmjdyn}{\IMJCCR}} \formspace.
\end{equation}
We will sometimes equivalently write $\AA_{m,j} = {\Quotientscale{\BUOmega}{\IMJ}}$, where $\IMJ$ is the two-sided ideal generated by both of the wanted relations. A smeared \emph{quantum Proca field} is then an element
\begin{equation}
\A_{m,j}(F) \coloneqq \big[\big(0,F,0,0,\dots\big)\big]_{m,j} \in \AA_{m,j} \formspace,
\end{equation}	
where the equivalence class $[\PH]_{m,j}$ is taken w.r.t. $\IMJ$. By construction, the quantum Proca fields fulfill the desired dynamical and commutation relations.
We can endow $\AA_{m,j}$ with the locally convex quotient topology obtained from $\BUOmega$ (cf. \cite[Theorem 12.14.8]{Dieudonne1970}), which is induced by the semi-norms
\begin{equation}
q_{m,j, \alpha}( [f]_{m,j} ) = \inf\big\{ p_\alpha(g) : g \in [f]_{m,j} \big\}
\end{equation}
where $\left\{ p_\alpha\right\}_\alpha $ is a family of semi-norms on $\BUOmega$ that induces its topology \cite[Lemma 12.14.8]{Dieudonne1970}. Note that the multiplication in $\AA_{m,j}$ is again jointly continuous\footnote{To ensure that the quotient space is Hausdorff, we will show below that the ideals  $\IMJDYN$ and $\IMJCCR$ are closed.}.

\subsubsection{Reduction to the current-free case}\label{sec:no-current}

We now show that the algebra $\AA_{m,j}$ with source dependent dynamics is homeomorphic to the algebra $\AA_{m,0}$ with vanishing source, where the subscript $0$ indicates that we set $j=0$.

Let us fix a solution $\varphi$ of the classical source dependent Proca equation, $(\delta d +m^2) \varphi = j$. We may then define a *-algebra-homomorphism $\Gamma_{\varphi}$ on $\BUOmega$ which preserves the unit and which is then uniquely determined by its action on homogeneous elements of degree one:
	\begin{equation}
	\Gamma_{\varphi} :
	\big(0,F,0,0,\dots\big) \mapsto \big(-\langle \varphi , F \rangle_\M , F , 0,0,\dots\big)
	\end{equation}
for all $F \in \Omega^1_0(\M)$.

\begin{theorem}\label{thm:field-algebra-homeomorphy}
	Let $m > 0$ and $j \in \Omega^1(\M)$ and $\varphi\in\Omega^1(\M)$ a solution of $(\delta d +m^2) \varphi = j$. Then the map $\Gamma_{\varphi}$ is a homeomorphism of $\BUOmega$ which descends to a homeomorphism $\Psi_{\varphi} : \AA_{m,0} \to \AA_{m,j}$.
\end{theorem}
\begin{proof}
	The inverse $\Gamma_{\varphi}$ is obviously determined by
	\begin{align}
	\Gamma_{\varphi}^{-1} : \BUOmega &\to \BUOmega \\
	\big(0,F,0,0,\dots\big) &\mapsto \big(+\langle \varphi , F \rangle_\M , F , 0,0,\dots\big) \formspace\notag
	\end{align}
	and both $\Gamma_{\varphi}$ and $\Gamma_{\varphi}^{-1}$ are continuous on $\BUOmega$. We now show that $\Gamma_{\varphi}$ maps the ideal $\IMZ$ onto $\IMJ$. It suffices to show that the generators of the source-free ideal map under $\Gamma_{\varphi}$ to the corresponding generators of the source dependent ideal and vice versa. Let $F\in \Omega^1_0(\M)$, then
	\begin{align}
	\Gamma_{\varphi}&\Big( \big( 0,(\delta d +m^2)F , 0 , 0 , \dots   \big)  \Big) \notag\\
	&= \big( - \langle \varphi, (\delta d + m^2)F \rangle_\M,(\delta d +m^2)F , 0 , 0 , \dots   \big) \notag\\
	&= \big( - \langle (\delta d + m^2) \varphi, F \rangle_\M,(\delta d +m^2)F , 0 , 0 , \dots   \big) \notag\\
	&= \big( - \langle j, F \rangle_\M ,(\delta d +m^2)F , 0 , 0 , \dots   \big)  \formspace,
	\end{align}
	so the generators for the dynamics transform in the desired way. For the commutation relations we first decompose:
	\begin{align}
	\big( -\i \Gm{F}{F'},0 , F \otimes F' - F' \otimes F   , 0 ,0, \dots   \big)  &=\big( -\i \Gm{F}{F'}, 0 ,0, \dots   \big) \\
	&\phantom{M}+ \big( 0,F,0,0,\dots  \big)\cdot \big( 0,F',0,0,\dots  \big)  \notag \\
	&\phantom{M}-\big( 0,F',0,0,\dots  \big)\cdot \big( 0,F,0,0,\dots  \big) \notag
	\end{align}
	and therefore obtain
	\begin{align}
	\Gamma_{\varphi}\Big( \big( -\i \Gm{F}{F'},&0 , F \otimes F' - F' \otimes F   , 0 ,0, \dots   \big)  \Big) \notag\\
	&=\big( -\i \Gm{F}{F'}, 0 ,0, \dots   \big)  \notag \\
	&\phantom{M}+ \big( - \langle \varphi,F \rangle_\M,F,0,0,\dots  \big)\cdot \big( - \langle \varphi,F' \rangle_\M,F',0,0,\dots  \big)  \notag \\
	&\phantom{M}-\big( - \langle \varphi,F' \rangle_\M,F',0,0,\dots  \big)\cdot \big( - \langle \varphi,F \rangle_\M,F,0,0,\dots  \big) \notag\\		
	&= \big( -\i \Gm{F}{F'},0 , F \otimes F' - F' \otimes F   , 0 ,0, \dots   \big)   \formspace.
	\end{align}
	It is straightforward to check in a completely analogous fashion that the generators of the source-dependent ideal map under $\Gamma_{\varphi}^{-1}$ to the generators of the source-free ideal.
	In conclusion, we find that $\Gamma_{\varphi}(\IMZ)=\IMJ$, and diving out the ideals yields the diffeomorphism $\Psi_{\varphi}$. We refer to \cite[Theorem 4.15]{Schambach2016} for more details.
\end{proof}
Given an observable of the source free theory $\A_{m,0}(F)$, we obtain
\begin{equation}
\A_{m,j}(F)
 = \langle \varphi, F \rangle_\M \cdot \mathbbm{1}_{\AA_{m,j}} + \Psi_{\varphi}\big(\A_{m,0}(F)\big) \formspace.
\end{equation}
Hence, the dynamics and commutation relations for $\A_{m,0}$ imply those of $\A_{m,j}$ and vice versa.

\subsubsection{Initial value-formulation}\label{sec:initial-value}

In order to divide out the dynamical ideal $\IMZDYN$ in the source-free case it is convenient to make use of an initial value formulation. First, however, we characterise the generators of this ideal:
\begin{lemma}\label{lem:dynamics}
$F\in\Omega^1_0(\M)$ is of the form $F=(\delta d+m^2)F'$ for some $F'\in\Omega^1_0(\M)$ if and only if $G_mF=0$.
\end{lemma}
\begin{proof}
If $F=(\delta d + m^2)F'$, then $G_mF=G_m(\delta d+m^2)F'=0$. Conversely, if $G_mF=0$, then $F'=G_m^+F=G_m^-F$ has compact support and $F=(\delta d+m^2)F'$.
\end{proof}
Now let $\Sigma$ be an arbitrary, fixed Cauchy surface. We will use the short-hand notation $\Dzs = \Omega^1_0(\Sigma) \oplus \Omega^1_0(\Sigma)$ for the space of initial data on $\Sigma$. We define the map
\begin{equation}
\kappa_m: \Omega^1_0(\M) \to \Dzs  \formspace,\quad
F \mapsto (\rhoz G_m F , \rhod G_m F) \formspace,
\end{equation}
which maps a test one-form $F$ to the solution $G_mF$ of Proca's equation and then to its initial data on $\Sigma$ (cf.~Theorem \ref{thm:solution_proca_unconstrained} and Definition \ref{def:cauchy_mapping_operators}). In the notation, we omit the dependence of the map on the Cauchy surface.

For any value of $m>0$, $\kappa_m$ is continuous w.r.t.~the direct sum topology on $\Dzs$, and hence $\KERN{\kappa_m}$ is closed \cite[34-36 ]{Treves1967}. 
By Lemma \ref{lem:dynamics} and Theorem \ref{thm:solution_proca_unconstrained} we have
\begin{equation}\label{eqn:ker(kappa_m)}
\KERN{\kappa_m} = \big\{F \in \Omega^1_0(\M) \mid G_mF=0 \big\} = (\delta d + m^2)\Omega^1_0(\M) \formspace .
\end{equation}

By a standard construction \cite[ibid.]{Treves1967}, illustrated in Diagram \ref{dia:homeomorphism_one_particle_level}, $\kappa_m$ gives rise to a linear map  $\xi_m : {\Quotientscale{\Omega^1_0(\M)}{\KERN{\kappa_m}}} \to \IMG{\kappa_m}$, which is the unique bijective map such that $\xi_m([F]_m) = \kappa_m(F)$, where $[F]_m$ denotes equivalence classes in the quotient space \cite[16]{Treves1967}. We will now show

\begin{table}
	\begin{displaymath}
	\xymatrix @R=20mm @C=30mm
	{
		\Omega^1_0 (\M)  \ar[r]^{\kappa_m}   \ar[dr]_{[{\PH}]_m}  			&     \IMG{\kappa_m} 													\ar@<-.5ex>[d]_{\xi_m^{-1}} \ar@{^{(}->}[r]^i &   \Dzs \\
		&      {{\Quotientscale{\Omega^1_0(\M)}{\KERN{\kappa_m}}}}          \ar@<-.5ex>[u]_{\xi_m^{\phantom{-1}}}
	}	
	\end{displaymath}
	\caption{Illustrating the construction of the homeomorphism $\xi_m$ of the space of dynamical test one-forms and the space of initial data.}
	\label{dia:homeomorphism_one_particle_level}
\end{table}
\begin{lemma}\label{lem:one-particle-homeomorphism}
	$\xi_m$ is a homeomorphism onto $\Dzs$.
\end{lemma}
\begin{proof}
	First we will show that $\kappa_m$ is surjective, by constructing a map $\vartheta_m : \Dzs \to \Omega^1_0(\M)$ such that $\xi_m \comp [\PH]_m \comp \vartheta_m = \mathrm{id}$.

We choose a fixed $\chi \in \Omega^0(\M)$ such that $\chi= 1$ on $J^+(\Sigma_+)$ and $\chi\equiv 0$ on $J^-(\Sigma_-)$, where $\Sigma_\pm$ are Cauchy surfaces in the future (+) and past (-) of $\Sigma$. Now let $(\varphi,\pi) \in \Dzs$ specify initial data on a Cauchy surface $\Sigma$. Then, by Theorem \ref{thm:solution_proca_unconstrained}, there exists a unique solution $A \in \Omega^1(\M)$ to the source free Proca equation $(\delta d + m^2) A = 0$ with the given data. We note that $\supp{A} \subset J\big( \supp{\varphi} \cup \supp{\pi} \big)$ (see \cite[Theorem 3.2.11]{BaerGinouxPfaeffle2007}) and hence, by defining
	\begin{equation}
	\vartheta_m(\varphi,\pi) \coloneqq -(\delta d + m^2) \chi A,
	\end{equation}
	we see that $\vartheta_m(\varphi,\pi)$ is a compactly supported one-form with support contained in the compact set $J\big( \supp{\varphi} \cup \supp{\pi}\big) \cap  J^-(\Sigma_+) \cap J^+(\Sigma_-)$.
	We want to show that $\kappa_m \vartheta_m(\varphi,\pi) = (\varphi,\pi)$. For this we note that the domains of $G_m^\pm$ can be extended to forms with past (+) resp.~future (-) compact supports \cite{Sanders2013,Schambach2016}. With these extended definitions we find
	\begin{align}
	G_m^+ \vartheta_m(\varphi,\pi)
	&= - \chi A \notag ,\\
	G_m^- \vartheta_m(\varphi,\pi)
	&= G_m^- (\delta d + m^2)(1-\chi) A \notag\\
	&= (1-\chi) A \formspace,
	\end{align}
because $\vartheta_m(\varphi,\pi)=(\delta d+m^2)(1-\chi) A$. We therefore find the result
	\begin{align}
	G_m \vartheta_m(\varphi,\pi)
	&= (G_m^- - G_m^+)\vartheta_m(\varphi,\pi) \notag\\
	&= (1- \chi) A + \chi A = A
	\end{align}
	and hence $\kappa_m\vartheta_m(\varphi,\pi)=(\varphi,\pi)$, which completes the proof of surjectivity. That is, we have found $\IMG{\kappa_m} = \Dzs$.\par
	It remains to show that the bijection $\xi_m$ is a homeomorphism. By construction, $\xi_m$ is continuous because $\kappa_m$ is continuous \cite[Proposition 4.6]{Treves1967}. The inverse is given by $\xi_m^{-1} = [\PH]_m \comp \,\vartheta_m$, where $\vartheta_m$ is continuous, because $A$ depends continuously on the initial data $(\varphi,\pi)$. Since $[\PH]_m$ is also continuous, so is $\xi^{-1}_m$.	This completes the proof.
\end{proof}
We will now generalize these ideas to the algebra $\BUOmega$ in order to implement the dynamics by dividing out the ideal generated by $\big(0,(\delta d + m^2)F,0,0,\dots\big)$,  where $F\in \Omega^1_0(\M)$. As we did on the degree-one level, we would like to find a map $K_m : \BUOmega \to \BU\big(\Dzs\big)$ such that $\KERN{K_m} = \IMZDYN$ and then show that $\BUmzdyn$ is homeomorphic to $\BU\big(\Dzs\big)$. We do this by \emph{lifting} the map $\kappa_m$ to the BU-algebra:
We define $K_m : \BUOmega \to \BU\big(\Dzs\big) $ as a BU-algebra-homomorphism which preserves the units and which is then completely determined by its action on homogeneous degree-one elements:
\begin{align}
K_m : \BUOmega &\to \BU\big(\Dzs\big) \\
\big(0,F,0,0,\dots\big) &\mapsto 	\big(0, \kappa_m (F),0,0,\dots\big)\notag \formspace.
\end{align}
With this map we can, analogously to the degree-one-part, construct a homeomorphism that implements the dynamics.
\begin{lemma}\label{lem:field-algebra-homeomorphism}
	Let $m>0$ and $j=0$. Then the map $K_m : \BUOmega \to \BU\big(\Dzs\big)$ descends to a homeomorphism $\Xi_m : \BUmzdyn \to \BU\big( \Dzs \big)$ with $\Xi_m \big([f]_{m,0}^\mathrm{dyn}\big) = K_m(f)$ where $f \in \BUOmega$.
\end{lemma}
\begin{proof}
	The surjectivity of $K_m$ follow directly from the surjectivity of $\kappa_m$, which was established in the proof of Lemma \ref{lem:one-particle-homeomorphism}. Because $\kappa_m$ is continuous, so is $\kappa_m^{\otimes n}$ on $\Gamma_0(T^*\M^{\boxtimes n})$ for any $n\ge 1$, by Schwartz' Kernels Theorem. Therefore, $\kappa_m^{\otimes n}$ is also continuous on the algebraic tensor product $\big(\Omega^1_0(\M)\big)^{\otimes n}$ and hence $K_m$ is continuous. It follows that $K_m$ descends to a continuous linear map
$\Xi_m : {\Quotientscale{\BUOmega}{\KERN{K_m}}}\to \BU\big( \Dzs \big)$ (cf. \cite[Proposition 4.6]{Treves1967}). The inclusion $\IMZDYN \subset \KERN{K_m}$ is obvious from the facts that $K_m$ is an algebra homomorphisms and that the generators of $\IMZDYN$ are of the form $\big(0,F_i,0,0,\dots \big)$ with $F_i\in\KERN{\kappa_m}$, cf.~Equation.(\ref{eqn:ker(kappa_m)}). The non-trivial part is to show the converse inclusion $\KERN{K_m} \subset \IMZDYN$.\par 
Consider and arbitrary $f = \big(f^{(0)}, f^{(1)}, f^{(2)}, \dots , f^{(N)} , 0 , 0 ,\dots\big) \in \KERN{K_m}$, $f^{(k)} \in \left(\Omega^1_0(\M)\right)^{\otimes k}$. Because $K_m$ preserves degrees, each homogeneous element $\big(0,\dots ,0 , f^{(n)} , 0 ,0, \dots\big)$ is in $\KERN{K_m}$, and it suffices to prove that these homogeneous elements are in the ideal $\IMZDYN$. We will show by induction in the degree $n$ that an arbitrary homogeneous element $\big(0,\dots ,0,  f^{(n)} , 0 ,0, \dots\big)$ with $\kappa_m^{\otimes n}\left(f^{(n)}\right)  = 0$ is in the ideal $\IMZDYN$.\par 
At degree 0, $\kappa_m^{\otimes 0}$ is the identity mapping, so its kernel is trivial. At degree 1, we use the fact that $\kappa_m(F)=0$ if and only if $\big(0,F,0,0,\dots\big)$ is a generator of $\IMZDYN$ (cf.~Equation (\ref{eqn:ker(kappa_m)})).\par 
We can now make the induction step and assume that the claim holds for homogeneous elements of degree $\le n$ for some $n\ge 1$.
Consider a homogeneous element $\big(0,\dots,0,f^{(n+1)},0,0,\dots\big)$ where $f^{(n+1)} \in \left(\Omega^1_0(\M)\right)^{\otimes (n+1)}$ such that $\kappa_m^{\otimes(n+1)}(f^{(n+1)}) = 0$. We can write this more explicitly for some $F_i \in \Omega^1_0(\M)$ and some $\F_i^{(n)}\in \big(\Omega^1_0(\M)\big)^{\otimes n}$ as
	\begin{equation}
	\big(0,\dots,0,f^{(n+1)},0,0, \dots\big)  = \big(0,\dots,0, \sum\limits_{i=1}^{M} F_i \otimes \F_i^{(n)},0, 0 , \dots\big) \formspace.
	\end{equation}
	Let $V \coloneqq \SPAN{F_1,F_2, \dots , F_M}$ and $W \coloneqq V \hspace{0.005em}\cap\hspace{0.01em} \KERN{\kappa_m}$, which define finite dimensional subspaces of $\Omega^1_0(\M)$. We find a basis $\{ \widetilde{F}_1, \dots ,\widetilde{F}_\mu \}$, $\mu \leq M$, of $W$ which we can extend to a basis $\{ \widetilde{F}_1, \dots ,\widetilde{F}_M \}$ of $V$.
	With the use of this basis we can re-write
	\begin{align}
	f^{(n+1)} = \sum\limits_{i=1}^{M} F_i \otimes \F_i^{(n)}
	&= \sum\limits_{i=1}^{\mu} \widetilde{F}_i \otimes \widetilde{\F}_i^{(n)}  +  \sum\limits_{i=\mu + 1}^{M} \widetilde{F}_i \otimes \widetilde{\F}_i^{(n)} \notag\\
	& \eqqcolon X_1^{(n+1)} + X_2^{(n+1)} \formspace.
	\end{align}
	Here, each $\widetilde{\F}_i^{(n)}$ can be constructed as a linear combination of the ${\F}_i^{(n)}$'s. We first have a look at $X_1^{(n+1)}$. We know by construction for $i=1,\dots,\mu$ that $\kappa_m ( \widetilde{F_i} ) = 0$. It follows that
\begin{equation}
\big(0,\ldots,0,X_1^{(n+1)},0,\dots\big) = \sum\limits_{i=1}^{\mu} \big(0,\widetilde{F}_i,0,0, \dots\big) \otimes \big(0,\dots,0,\widetilde{\F}_i^{(n)},0,0, \dots\big)
\end{equation}
is in $\IMZDYN$ and that $\kappa_m^{\otimes(n+1)}(X_1^{(n+1)}) = 0$. Now we have a closer look at the remaining part $X_2^{(n+1)}$, which must then also have $\kappa_m^{\otimes(n+1)}(X_2^{(n+1)}) = 0$. However, by construction, it holds $\SPAN{ \widetilde{F}_{\mu+1}, \dots ,\widetilde{F}_M } \cap \KERN{\kappa_m} = \{ 0 \}$, which implies that the $\kappa_m(\widetilde{F_i})$'s are linearly independent for $i=\mu+1,\dots,M$. With this, it then follows that we must have $\kappa_m^{\otimes n}( \widetilde{\F}_i^{(n)}) = 0$ for all $i=\mu+1,\dots,M$. Since $\widetilde{\F}_i^{(n)}$ is of degree $n$, we can apply the induction hypothesis and find that
$(0,\dots,0,  \widetilde{\F}_i^{(n)} , 0 , 0 , \dots) \in \IMZDYN$ and hence
\begin{equation}
\big(0,\ldots,0,X_2^{(n+1)},0,\dots\big) = \sum\limits_{i=\mu+1}^{M} \big(0,\widetilde{F}_i,0,\dots\big) \otimes \big(0,\dots,0,\widetilde{\F}_i^{(n)},0,\dots\big)
\end{equation}
is also in $\IMZDYN$. Hence
\begin{equation}
\big(0,\dots, 0, f^{(n+1)} , 0 , 0 , \dots \big) = 	\big(0,\dots, 0, X_1^{(n+1)} + X_2^{(n+1)} , 0 , 0 , \dots \big) \in \IMZDYN
\end{equation}
which completes the proof by induction.
\end{proof}
The continuity of $K_m$ and the proof above imply in particular that the ideal $\IMZDYN=\KERN{K_m}$ is closed.

\subsubsection{Canonical commutation relations}\label{sec:CCR}

We are left to include the quantum nature of the fields by dividing out the relation that implements the CCR. In $\BUmzdyn$, we need to divide out the two-sided ideal $\IMJZCCR$ that is generated by elements $\big(-\i \Gm{F}{F'}, 0 , F \otimes F' - F' \otimes F , 0 , 0 , \dots\big)$. For $\BU\big( \Dzs \big)$ we make use of the following lemma:
\begin{lemma}\label{lem:propagator-non-degenerate}
	Let $F,F' \in \Omega^1_0(\M)$ be two test one-forms on $\M$ and let $(\varphi,\pi)=\kappa_m(F)$ and $(\varphi',\pi')=\kappa_m(F')$. Then
\begin{equation}
\Green{F}{F'} = \mathcal{G}^{(\Sigma)}\big(\kappa_m(F),\kappa_m(F')\big) \formspace,
\end{equation}where
	\begin{equation}
	\mathcal{G}^{(\Sigma)}\big((\varphi,\pi),(\varphi',\pi')\big)
	= \langle \varphi, \pi' \rangle_\Sigma  - \langle \pi , \varphi' \rangle_\Sigma   \formspace
	\end{equation}
is a symplectic form on the space $\Dzs$ of initial data, i.\,e., it is bilinear, anti-symmetric and non-degenerate.
\end{lemma}
\begin{proof}
	It is straightforward to show that $\mathcal{G}^{(\Sigma)}$ is a symplectic form. Now let $F,F' \in \Omega^1_0(\M)$ and recall that $G_m F'$ is a solution to the source free Proca equation with initial data $\kappa_m(F)$, and similarly for $G_mF'$. Then, using the definition of $\Green{F}{F'}$ and $\langle F , G_m^\pm F' \rangle_\M=\langle G_m^\mp F , F' \rangle_\M$,
	\begin{align}
	\Green{F}{F'}
	&= \langle F , G_m F' \rangle_\M
	= - \langle G_m F, F' \rangle_\M  \formspace \notag\\
	&= \langle \rhoz G_m F, \rhod G_mF' \rangle_\Sigma - \langle \rhod G_m F, \rhoz G_m F' \rangle_\Sigma\label{eqn:CCRdata} \notag\\
    &= \mathcal{G}^{(\Sigma)}\big((\varphi,\pi),(\varphi',\pi')\big)
	\end{align}
	by Theorem \ref{thm:solution_proca_unconstrained} with $j=0$.
\end{proof}
It follows from this lemma that $\IMJZCCR$ maps under $\Xi_m$ to the two-sided ideal $\ICCR \subset \BU\big( \Dzs \big)$ that is generated by elements \vspace{-.3cm}
\begin{equation}
\big(-\i ( \langle \varphi , \pi' \rangle_\Sigma - \langle \pi, \varphi'\rangle_\Sigma), 0 , (\varphi, \pi) \otimes (\varphi', \pi') - (\varphi', \pi') \otimes (\varphi, \pi) , 0 , 0 , \dots\big)\formspace.
\end{equation}
Lemma \ref{lem:symmetrization-of-fields} in Appendix \ref{app:lemmata} shows that the ideal $\ICCR$ is the kernel of a continuous linear map. This implies in particular that $\ICCR$, and hence also $\IMJZCCR$, is closed.

With these results the following theorem follows easily.
\begin{theorem}\label{thm:field_algebra_homeomorphy_source_free}
	Let $m>0$ and $j=0$. Then the map $\Xi_m : \BUmzdyn \to \BU\big( \Dzs \big)$ descends to a homeomorphism
$\Lambda_m : \AA_{m,0}(M) \to \Quotientscale{\BU\big( \Dzs \big)}{\ICCR}$.
\end{theorem}
We omit the proof and refer to \cite[Theorem 4.14]{Schambach2016} for the details.

The results of this section can be combined with those of Section \ref{sec:no-current} and \ref{sec:initial-value} and illustrated as in Diagram \ref{dia:final_structure}.
\begin{table}[]
	\begin{displaymath}
	\xymatrix @R=20mm @C=30mm
	{ 																																	& \BU{\left( \Dzs \right)} 		\ar[r]^{[{\PH}]_{\sim}^\text{CCR}}		\ar@<.5ex>[d]^{\Xi^{-1}_{m}}				& {{\Quotientscale{\BU{\left( \Dzs \right)}}{\ICCR}}} \ar@<.5ex>[d]^{\Lambda^{-1}_m}\\
		\BUOmega \ar[r]^{[{\PH}]_{m,0}^\text{dyn}}  	\ar@<.5ex>[d]^{\Gamma_{\varphi}^{\phantom{-1}}} \ar[ur]^{K_m}	&  \BUmzdyn  \ar[r]^{[{\PH}]_{m,0}^\text{CCR}} \ar@<.5ex>[u]^{\Xi_{m}^{\phantom{-1}}}  \ar@<.5ex>[d] &   \AA_{m,0} \ar@<.5ex>[d]^{\Psi_{\varphi}^{\phantom{-1}}} \ar@<.5ex>[u]^{\Lambda_m^{\phantom{-1}}}  \\
		\BUOmega  \ar[r]_{[{\PH}]_{m,j}^\text{dyn}}      \ar@<.5ex>[u]^{\Gamma^{-1}_{\varphi}}    & \BUmjdyn \ar[r]_{[{\PH}]_{m,j}^\text{CCR}}	\ar@<.5ex>[u]& 	\AA_{m,j} 	\ar@<.5ex>[u]^{\Psi_{\varphi}^{-1}}     \\
	}				
	\end{displaymath}
	\caption{A commutative diagram illustrating the various quotients of BU-algebras and their relations. Bi-directional arrows represent homeomorphisms.}
	\label{dia:final_structure}
\end{table}

\subsection{Locality of the quantum Proca field}\label{sec:LCQFT}

Finally we consider the quantum Proca field in the generally covariant setting, using a categorical framework as Brunetti, Fredenhagen and Verch \cite{BrunettiFredenhagenVerch2003}. For this purpose we introduce the following
\begin{definition}\label{def:categories_alg_spaccurr}
	By an \emph{admissible embedding} $\psi : (\M,g_\M) \to (\N,g_\N)$ we mean an orientation and time orientation preserving isometric embedding $\psi:\M\to\N$ such that for every $p \in \M$ it holds $J_\M^\pm(p) = \psi^{-1} \big( J_\N^\pm ( \psi(p)) \big)$.\par
	The category $\SpacCurr$  consists of triples $M=(\M,g_\M,j_\M)$ as objects, where $(\M,g_\M)$ is a (oriented and time-oriented) globally hyperbolic spacetime  and $j_\M \in \Omega^1(\M)$ is a background current,
	and morphisms $\psi$, where $\psi$ is an admissible embedding such that $\psi^* j_\N = j_\M$. \par
	The category $\Alg$ consists of unital $^*$-algebras as objects and unit preserving  $^*$-algebra-homomorphisms as morphisms. \par
	The category $\Alg'$ is the subcategory of $\Alg$ consisting of the same objects but only injective morphisms.
\end{definition}
\begin{definition}\label{def:generally-coveriant-qftcs}
	A \emph{generally covariant quantum field theory with background source} is a covariant functor between the categories $\SpacCurr$ and $\Alg$. The theory is called \emph{locally covariant} if and only if the range of the functor is contained in $\Alg'$.
\end{definition}
The construction of this functor $\mathbf{A}_m$ for the Proca field of mass $m>0$ is straightforward: To each $M$ we associate the $^*$-algebra $\mathbf{A}_m(M)\coloneqq \AA_{m,j}$ constructed on $M$ as above and to any morphism $\psi:M\to N$ we associate the unit preserving $^*$-algebra-homomorphism $\mathbf{A}_m(\psi) \equiv \alpha_\psi : \mathbf{A}_m(M) \to \mathbf{A}_m(N)$, whose action is fully determined by the action on the generators $\A_{m,M}(F)$,  which we previously denoted by $\A_{m,j}(F)$ without explicitly referring to the background spacetime $M$, as 
	\begin{equation}
	\alpha_\psi \big(\A_{m,M}(F)\big) = \A_{m,N}(\psi_*(F)) \formspace.
	\end{equation}
It is straightforward to show that the above functor is well-defined for all $m>0$. A detailed verification is given in \cite{Schambach2016}.

We now show that for $m>0$ the functor $\mathbf{A}_m$ defines a \emph{locally} covariant QFT, i.\,e.~that the homomorphisms $\mathbf{A}_m(\psi) \equiv \alpha_\psi$ are injective.

\begin{theorem}
	$\mathbf{A}_m$ as given above defines a locally covariant QFT of the Proca field, i.\,e.~it is a functor
$\mathbf{A}_m : \SpacCurr \to \Alg'$.
\end{theorem}
\begin{proof}
    $\mathbf{A}_m$ is given as a functor into $\Alg$, so it only remains to show that the morphisms $\mathbf{A}_m(\psi) \equiv \alpha_\psi$ are injective. By Lemma \ref{lem:propagator-non-degenerate} $\mathcal{G}^{(\Sigma)}$ is a symplectic form on $\Dzs$ and hence the algebra $\Quotientscale{\BU\big( \Dzs \big)}{\ICCR}$ is simple (cf. \cite[Scholium 7.1]{BaezSegalZhou1992}). The same is true for the homeomorphic algebra $\mathbf{A}_m(M)$ (cf. Theorems \ref{thm:field_algebra_homeomorphy_source_free} and \ref{thm:field-algebra-homeomorphy}). Since $\mathbf{A}_m(M)$ is simple, the homomorphism $\alpha_\psi$ has either full or trivial kernel. As $\alpha_\psi$ is defined to be unit preserving, it follows that the kernel is trivial and hence $\alpha_\psi$ is injective.
\end{proof}

\section{The zero mass limit}\label{sec:limit}
For the main results of this article we will investigate the zero mass limit of the Proca field a in curved spacetime in both the classical and the quantum case. In Section \ref{sec:continuity} we will formulate the key notion of continuity of the field theory with respect to the mass and establish its basic properties. We then define the massless limit in a general, state independent setup first for the classical Proca field in Section \ref{sec:zero-mass-limit-classical} and then for the quantum Proca field in Section \ref{sec:mass_dependence_and_limit}. At given points, we compare our results with the theory of the (quantum) vector potential of electromagnetism in curved spacetimes as studied in \cite{SandersDappiaggiHack2014, Pfenning2009}.\par

\subsection{Continuity in the mass}\label{sec:continuity}

When defining a notion of continuity of the field theory with respect to the mass, the basic problem is that at different masses the smeared fields $\A_{m,j}(F)$ are elements of different algebras $\AA_{m,j}$. Indeed, when constructing $\AA_{m,j}$ as a quotient of the BU-algebra, the ideals that implement the dynamics and the commutation relations both depend on the mass. We therefore need to find a way of comparing the Proca fields at different masses with each other.

One could try to solve this using the $C^*$-Weyl algebra to describe the quantum Proca field and the notion of a continuous field of $C^*$-algebras depending on the mass parameter (cf.~\cite{BinzHoneggerRieckers2004}). This would work very nicely, if the theories were described by a weakly continuous family of (non-degenerate) symplectic forms on a fixed linear space (cf.~\cite[Appendix A]{Schambach2016}, which generalises \cite{BinzHoneggerRieckers2004}). However, as it turns out, this approach is ill-suited for the problem at hand. Indeed, one would like linear combinations of Weyl operators
\begin{equation}
W_{m,j}(F_i)=\e^{\i\A_{m,j}(F_i)}
\end{equation}
with fixed test-forms $F_i\in\Omega_0^1(\mathcal{M})$ to depend continuously on the mass, but for $j=0$ the norm of an operator like $W_{m,0}\big((\delta d +m_0^2)F\big)-1$, with a fixed $F$ and $m_0$, can be seen to be discontinuous at $m=m_0$, where the operator vanishes.

A different attempt, which we have hinted at in Section \ref{sec:BU-algebra}, is to use the semi-norms
\begin{equation}
q_{m, j, \alpha}\big( [f]_{m,j} \big) = \inf\big\{ p_\alpha(g) : g \in [f]_{m,j} \big\}
\end{equation}
to define a notion of continuity of the theory with respect to the mass $m$. We could call a family of operators $\left\{O_m\right\}_{m>0}$ with $O_m \in \AA_{m,j}$ continuous if and only if the map $m \mapsto q_{m, j, \alpha}\big(O_m\big)$ is continuous for all $\alpha$ with respect to the standard topology in $\IR$.  While this definition seems appropriate at first sight, it is non-trivial to show the desirable property that for a fixed $F\in\Omega_0^1(\mathcal{M})$ the smeared field operators $\A_{m,j}(F)$ vary continuously with $m$. Even for $j=0$ and considering only the one-particle level, we were unable to prove this.

In this paper we therefore opt for the following solution, which makes use of the Borchers-Uhlmann algebra of initial data. For simplicity we first consider the case $j=0$ and a family of operators $\left\{O_m\right\}_{m>0}$ with $O_m \in \AA_{m,0}$. Since we have found for every mass $m>0$ that $\AA_{m,0}$ is homeomorphic to $\Quotientscale{\BU\big(\Dzs\big)}{\ICCR}$, we can map the family $\left\{O_m\right\}_{m>0}$ to a family of operators in the single algebra ${\Quotientscale{\BU\big(\Dzs\big)}{\ICCR}}$, which already carries a topology and hence a notion of continuity. When $j\not=0$ we combine this idea with the fact that $\AA_{m,j}$ is homeomorphic to $\AA_{m,0}$. In this way we arrive at the following notion of continuity.
\newpage
\begin{definition}[Continuity with respect to the mass]\label{def:field_continuity_general}
	Let $j \in \Omega^1(\M)$ be fixed and let $\left\{O_m\right\}_{m>0}$ be a family of operators with $O_m\in\AA_{m,j}$. We call $\left\{O_m\right\}_{m>0}$ continuous if and only if the map
	\begin{align}
	\IR_+ &\to {\Quotientscale{\BU\big( \Dzs \big) }{\ICCR}} \formspace,\\
	m &\mapsto \big( \Lambda_m \comp  \Psi_{\varphi_{m,j}}^{-1} \big) \left(O_m\right)\notag
	\end{align}
	is continuous, where $\Lambda_{m}$ and $\Psi_{\varphi_{m,j}}$ are as defined in Section \ref{sec:CCR} and \ref{sec:no-current} and $\left\{ \varphi_{m,j} \right\}_{m>0}$ is a family of classical solutions to the inhomogeneous Proca equation $(\delta d + m^2) \varphi_{m,j} = j$ which depends continuously on $m$ (i.\,e. $m\mapsto \varphi_{m,j}\in \Omega^1(\M)$ is continuous).\par 
	Equivalently, identifying $O_m = \big[\tilde{O}_m\big]_{m,j}$ for some $\tilde{O}_m \subset \BUOmega$, the family $\left\{O_m\right\}_{m>0}$ is continuous if and only if the map
	\begin{align}
	\IR_+ &\to {\Quotientscale{\BU\big( \Dzs \big) }{\ICCR}} \formspace,\\
	m &\mapsto \big[  \big( K_m \comp  \Gamma_{\varphi_{m,j}}^{-1} \big) (\tilde{O}_m) \big]_\sim^\text{CCR}\notag
	\end{align}
	is continuous, with $K_m$ and $\Gamma_{\varphi_{m,j}}$ as defined in Section \ref{sec:initial-value} and \ref{sec:no-current}.
\end{definition}
We now aim to establish some desirable properties of this notion of continuity, most importantly that it is independent of the choice of Cauchy surface $\Sigma$ and of the choice of the continuous family $\varphi_{m,j}$ of classical solutions. Our arguments will make essential use of the following result for normally hyperbolic operators:
\begin{theorem}\label{thm:Emcont}
Let $P$ be a normally hyperbolic operator on a real vector bundle $V$ over a globally hyperbolic spacetime $M$. Let $u_0,u_1\in\Gamma(V|_{\Sigma})$ be initial data on a Cauchy surface $\Sigma$ and $f\in\Gamma_0(V)$. For $r\in\mathbb{R}$, let $u^{(r)}$ be the unique solution to $(P+r)u^{(r)}=f$ with initial data $u_0,u_1$ on $\Sigma$. Then $r\mapsto u^{(r)}$ is a continuous map from $\mathbb{R}$ to $\Gamma(V)$.
\end{theorem}
\begin{proof}
It suffices to prove continuity at $r=0$, after shifting $P$ by a constant. We may write $P=\nabla^{\alpha}\nabla_{\alpha}+B$, where $B$ is a bundle endomorphism \cite{BaerGinouxPfaeffle2007}. Here,  $\nabla_{\alpha}$ is a connection on $V$, which may be extended with the Levi-Civita connection to tensor product bundles of $V$, $TM$ and their dual bundles. We write for $k=0,1,2,\ldots$
\begin{equation}\label{Eqn_Defvk}
v^{(k,r)}_{\alpha_1\cdots\alpha_k} \coloneqq \nabla_{\alpha_1}\cdots\nabla_{\alpha_k}(u^{(r)}-u^{(0)})
\end{equation}
and we note that $(P+r)(u^{(r)}-u^{(0)})=-ru^{(0)}$ and hence
\begin{equation}\label{Eqn_hyperbolicvk}
(P+r)v^{(k,r)}_{\alpha_1\cdots\alpha_k} = -r\nabla_{\alpha_1}\cdots\nabla_{\alpha_k}u^{(0)}
-(B^{(k)}v^{(k,r)})_{\alpha_1\cdots\alpha_k}+\sum_{l=0}^{k-1}(C^{(k,l)}v^{(l,r)})_{\alpha_1\cdots\alpha_k} \formspace,
\end{equation}
where $B^{(k)}$ and $C^{(k,l)}$ are bundle homomorphisms which involve $B$ and the curvature of $\nabla$. It follows that $v^{(k,r)}$ solves an inhomogeneous normally hyperbolic equation with the operator $P+B^{(k)}+r$ and an inhomogeneous term determined by $u^{(0)}$ and $v^{(l,r)}$ with $l<k$.

We now first prove by induction over $k$ that the initial data of $v^{(k,r)}$ converge to 0 in $\Gamma(V|_{\Sigma})$ as $r\to0$. For $k=0$ this claim is trivial, because $v^{(0,r)}=u^{(r)}-u^{(0)}$ has vanishing initial data for all $r$. Now suppose that the claim is true for all $0\le l\le k-1$ and consider $v^{(k,r)}_{\alpha_1\cdots\alpha_k}$. Using the unit normal vector field $n$ to $\Sigma$ we may express $v^{(k,r)}_{\alpha_1\cdots\alpha_k}$ as a sum of terms in which all indices are either projected onto the conormal direction or onto the space-like directions cotangent to $\Sigma$. If one of the indices is projected onto the space-like directions, then we may commute the derivatives in Equation (\ref{Eqn_Defvk}) to bring the space-like index to the left. The commutator terms involve the curvature, which is independent of $r$, and at most $k-2$ derivatives. Hence its initial data vanish as $r\to0$ by the induction hypothesis. Similarly, if the first index is space-like, then the initial data of the term vanish as $r\to 0$ by the induction hypothesis, since convergence in $\Gamma(V|_{\Sigma})$ entails convergence of all spacelike derivatives. Finally we consider the term where all indices are projected onto the conormal direction. For this term we may use Equation (\ref{Eqn_hyperbolicvk}) to eliminate two normal derivatives in favour of spacelike derivatives and lower order terms. Again the initial data of this term vanish in the limit $r\to0$ by the induction hypothesis. Adding all components together proves that the initial data of $v^{(k,r)}$ converge to 0 in $\Gamma(V|_{\Sigma})$ as $r\to0$.

At this point our proof uses an energy estimate. To formulate it, we endow the vector bundles $V$ and $TM$ with auxiliary smooth Riemannian metrics, and we denote the corresponding pointwise norms by $\|{\PH} \|$. For every compact $K\subset \Sigma$ and $L\subset\mathbb{R}$ there is a $C>0$ such that for all $r\in L$
\begin{equation}\label{Eqn_energyestimate}
\int_{D(K)}\norm{v^{(r)}}^2\le C\int_K \left(\norm{ \restr{v^{(r)}}{\Sigma}}^2
+ \norm{n^{\alpha}\nabla_{\alpha} \restr{v^{(r)}}{\Sigma}}^2\right)+C\int_{D(K)} \norm{f^{(r)}}^2 \formspace,
\end{equation}
where $D(k)$ is the domain of dependence and $v^{(r)}$ is a solution to\footnote{An explicit proof of this estimate is in Appendix \ref{app_energy_estimate}. Cf.~\cite[App.3, Thm.3.2]{Choquet-Bruhat} for an energy estimate of a quite similar form, where the independence of $C$ on $r$ can be established by retracing the steps in the proof.} $(P+r)v^{(r)}=f^{(r)}$.

We now apply this result to $T^*M^{\otimes k}\otimes V$ instead of $V$ and prove by induction that each $v^{(k,r)}$ converges to 0 in the $L^2$-sense on every compact set $\widetilde{K}\subset M$. Indeed, $\widetilde{K}\subset D(K)$ for some compact $K\subset\Sigma$, so it suffices to apply the above energy estimate to $v^{(k,r)}$ and show that the right-hand side converges to 0. Note that the initial data of $v^{(k,r)}$ converge to 0 in $\Gamma(V|_{\Sigma})$, and hence also in the $L^2$-norm on every compact $K$. It remains to consider the source term of Equation (\ref{Eqn_hyperbolicvk}),
\begin{equation}
-r\nabla_{\alpha_1}\cdots\nabla_{\alpha_k}u^{(0)}
+\sum_{l=0}^{k-1}(C^{(k,l)}v^{(l,r)})_{\alpha_1\cdots\alpha_k} \formspace.
\end{equation}
Because $u^{(0)}$ is independent of $r$ we see immediately that the first term converges to 0 as $r\to0$. For $k=0$ the summation vanishes, so the energy estimate proves the desired convergence of $v^{(0,r)}$. For $k>0$ we use a proof by induction. Assuming that $v^{(l,r)}\to 0$ in the $L^2$-sense as $r\to 0$ for all $0\le l\le k-1$, the energy estimate then proves the claim also for $v^{(k,r)}$.\par 
Finally, since $v^{(0,r)}$ and all its derivatives converge to 0 in an $L^2$ sense on every compact set, they also converge in $\Gamma(V)$ by the Sobolev Embedding Theorem (\cite[Sec.5.6 Theorem 6]{Evans} ).
\end{proof}
For us, the following consequence is most relevant:
\begin{corollary}\label{Cor:Emcont}
For fixed $F\in \Omega^1_0(\M)$, the advanced and retarded solutions $E^{\pm}_mF$ depend continuously on $m\in\mathbb{R}$. Consequently, $E_mF$ is continuous in $m\in\mathbb{R}$ and $G_m^\pm F$ and $G_mF$ are continuous in $m>0$.
\end{corollary}
\begin{proof}
We apply Theorem \ref{thm:Emcont} to $\Box+m^2$ with $r=m^2$. Choosing e.\,g.~$u_0,u_1=0$ and $\Sigma$ to the past/future of the support of $f$, we see that $E^{\pm}_mF$ depend continuously on $m\in\mathbb{R}$, and hence so does $E_mF$. The continuity of $G^{\pm}_mF$ and $G_mF$ follows from the formula $G_m^\pm = (m^{-2}{d\delta} +1) E_m^\pm$ as long as $m\not=0$.
\end{proof}
Let us now return to the continuity of families of observables and verify that it behaves well in the simplest examples.
\newpage
\begin{lemma}
For a fixed $F \in \Omega^1_0(\M)$ and $j\in\Omega^1(\M)$ the family of operators $\left\{\A_{m,j}(F)\right\}_{m>0}$ is continuous.
\end{lemma}
\begin{proof}
We see from the definitions of the maps involved in Definition \ref{def:field_continuity_general} that
\begin{equation}
\big(\Lambda_m \comp  \Psi_{\varphi_{m,j}}^{-1}\big)(\A_{m,j}(F)) =
\big[\big( \langle \varphi_{m,j} , F \rangle_{\M} , \kappa_mF,0,0,\dots\big)\big]_\sim^\text{CCR},
\end{equation}
where $[{\PH}]_\sim^\text{CCR}$ is continuous and does not depend on the mass. Because $\varphi_{m,j}$ depends continuously on $m>0$, so does
$\langle \varphi_{m,j} , F \rangle_{\M}$. Furthermore, $G_mF$ is continuous in $m>0$ by Corollary \ref{Cor:Emcont} and the operators $\rho_{({\PH})}$ are continuous and independent of $m$, therefore, the initial data $\kappa_m F=(\rhoz G_m F , \rhod G_m F)$ also depend continuously on $m>0$. Combining these continuous maps proves the lemma.
\end{proof}
We have found the desirable property that the quantum fields vary continuously with respect to the mass. Note that this result is in fact independent of the choice of the Cauchy surface $\Sigma$, since $\kappa_m(F)$ is continuous in $m$ for every Cauchy surface. Indeed, we will now show quite generally that the notion of continuity in Definition \ref{def:field_continuity_general} is independent of the choice of the Cauchy surface $\Sigma$ and of the family of classical solutions $\left\{\varphi_{m,j} \right\}_m$.
\begin{theorem}\label{thm:continuity-independence-source-free}
The notion of continuity in Definition \ref{def:field_continuity_general} is independent of the choice of the Cauchy surface $\Sigma$ and of the family $\left\{ \varphi_{m,j} \right\}_{m>0}$ of classical solutions to the inhomogeneous Proca equation.
\end{theorem}
\begin{proof}
In this proof we will make repeated use of a joint continuity lemma, which we state and prove as Lemma \ref{lem:jointcontinuity} in Appendix \ref{app:lemmata}. This lemma makes use of barrelled locally convex spaces, and we prove in Lemma \ref{lem:BU-algebra-barreled} that the complete BU-algebra is such a space.\par 
Let $\left\{ O_m\right\}_{m>0}$ be a family of operators with $O_m\in\AA_{m,j}$. We first verify the independence of the choice of Cauchy surface. For this we choose two Cauchy surfaces $\Sigma$, $\Sigma'$ and we consider the family of operators $O'_m \coloneqq \Psi_{\varphi_{m,j}}^{-1} (O_m)\in \AA_{m,0}$. It then suffices to prove that the continuity of $\Lambda^{(\Sigma)}_m(O'_m)$ implies the continuity of $\Lambda^{(\Sigma')}_m(O'_m)$, where we have made the dependence on the Cauchy surfaces explicit.\par 
Let us first consider the space of initial data on the Cauchy surface for the wave equation on one-forms, $\widetilde{\mathcal{D}}_0(\Sigma) \coloneqq \Omega^1_0(\Sigma)\oplus\Omega^1_0(\Sigma)\oplus\Omega^0_0(\Sigma)\oplus\Omega^0_0(\Sigma)$ and its analogue $\widetilde{\mathcal{D}}_0(\Sigma')$. For each $m$ we may define a continuous linear map $L_m:\widetilde{\mathcal{D}}_0(\Sigma)\to\widetilde{\mathcal{D}}_0(\Sigma')$, which propagates the initial data under the wave operator $\square+m^2$. By Theorem \ref{thm:Emcont}, $L_m$ is weakly continuous. \par 
Fixed initial data $\psi= (\varphi, \pi) \in \Dzs$ can be extended to initial data $\Psi_m\in\widetilde{\mathcal{D}}_0(\Sigma)$, using the constraint equations of Theorem \ref{thm:solution_proca_unconstrained} with $m>0$ and $j=0$. Note that $\Psi_m$ depends continuously on the mass $m$. Because $\widetilde{\mathcal{D}}_0(\Sigma)$ is a barrelled space (cf.~the proof of Lemma \ref{lem:BU-algebra-barreled}) we may apply Lemma \ref{lem:jointcontinuity} and conclude that $L_m\Psi_{m'}$ is jointly continuous in $(m,m')$ on $\IR_+\times\IR_+$. In particular, $m\mapsto L_m\Psi_m$ depends continuously on $m>0$. Consequently, the map $\tau_m : \Dzs \to \mathcal{D}_0(\Sigma')$, which propagates initial data for the Proca field of mass $m$, is also weakly continuous in $m>0$.\par 
We now extend this result as follows. For $1\le n\le N$ we consider the continuous linear map $T^{N,n}_m \coloneqq 1^{\otimes n-1}\otimes\tau_m\otimes 1^{\otimes N-n}$ on $\Gamma_0\big((\TsS\oplus\TsS)^{\boxtimes n}\boxtimes(\TsS'\oplus\TsS')^{\boxtimes N-n}\big)$, which may be defined using Schwartz' Kernels Theorem. One may extend the proof of Theorem \ref{thm:Emcont} and Corollary \ref{Cor:Emcont} to show that $T^{N,n}_m$ is also weakly continuous in $m>0$. We then define the map
$T^N_m \coloneqq T^{N,1}\comp T^{N,2}\comp \cdots \comp T^{N, N}$ which is again weakly continuous in $m>0$, by a repeated application of the joint continuity Lemma  \ref{lem:jointcontinuity}, using the fact that each of the spaces $\Gamma_0\big((\TsS\oplus\TsS)^{\boxtimes n}\boxtimes(\TsS'\oplus\TsS')^{\boxtimes N-n}\big)$ is barrelled. Let us now consider the lift of $\tau_m$ to a continuous linear map $T_m : \overline{\BU\big( \Dzs\big)}\to \overline{\BU\big( \mathcal{D}_0(\Sigma')\big)}$ between complete BU-algebras (using sections of the bundle $\TsS\oplus\TsS$ and its analogue on $\Sigma'$). Its action on $\Gamma_0\big((\TsS\oplus\TsS)^{\boxtimes n}\big)$ is simply given by $T^N_m$, which shows that $T_m$ is also weakly continuous in $m>0$. We note that each $T_m$ is a homeomorphism and that it maps the ideal $\ICCRS$ onto $\ICCRSP$. This means that it also maps the closed ideal $\overline{\ICCRS}$ onto $\overline{\ICCRSP}$ and it descends to a homeomorphism $\widetilde{T}_m$ between the quotient algebras ${\Quotientscale{\overline{\BU\big(\Dzs\big)}}{\;\overline{\ICCRS}}}$ and  ${\Quotientscale{\overline{\BU\big( \mathcal{D}_0(\Sigma')\big)}}{\; \overline{\ICCRSP}}}$. The weak continuity of $T_m$ in $m>0$ implies the weak continuity of $\widetilde{T}_m$ in $m>0$.\par 
The complete algebra $\overline{\BU\big(\Dzs\big)}$ is barrelled, as shown in Lemma \ref{lem:BU-algebra-barreled}, and hence so is the quotient ${\Quotientscale{\overline{\BU\big(\Dzs\big)}}{\;\overline{\ICCRS}}}$ \cite[Proposition 33.1]{Treves1967}. Furthermore, because the ideal $\ICCR$ is a closed subspace of $\BU\big(\Dzs\big)$ (cf.~Section \ref{sec:CCR}), the quotient ${\Quotientscale{\BU\big(\Dzs\big)}{\ICCRS}}$ is a dense subspace of ${\Quotientscale{\overline{\BU\big(\Dzs\big)}}{\;\overline{\ICCRS}}}$. On this subspace, $\widetilde{T}_m$ restricts to $\Lambda^{(\Sigma')}_{m} \comp \big( \Lambda^{(\Sigma)}_{m} \big)^{-1}$. Identifying
	\begin{equation}
	\Lambda^{(\Sigma')}_m(O'_m) = \widetilde{T}_m \Lambda^{(\Sigma)}_m(O'_m)
	\end{equation}
we may therefore use the assumed continuity of $\Lambda^{(\Sigma)}_m(O'_m)$ in $m>0$ and the known weak continuity of $\widetilde{T}_m$ together with Lemma \ref{lem:jointcontinuity} to find that $m\mapsto \Lambda^{(\Sigma')}_m(O'_m)$ is continuous in $m>0$. This proves the independence of the choice of $\Sigma$.\par 
%
We now turn to the independence of the choice of classical solutions. Let $\left\{\varphi_{m,j} \right\}_m$ and $\left\{\varphi'_{m,j} \right\}_m$ specify continuous families of classical solutions to the inhomogeneous Proca equation and fix a Cauchy surface $\Sigma$. We denote the initial data of $\varphi_{m,j}$ and $\varphi'_{m,j}$ by $(\phi_m,\pi_m)$ and $(\phi_m',\pi_m')$, respectively. For each $m>0$ we now define an algebra homeomorphism $L_m$ on the BU-algebra 
$\BU\big(\Dzs\big)$ by setting stipulating that $L_m$ preserves the unit and acts on homogeneous elements of degree 1 as
\begin{equation}
L_m\big(0,(\alpha,\beta),0,0,\ldots\big) \coloneqq \big(\mathcal{G}^{(\Sigma)}\big((\phi_m-\phi'_m,\pi_m-\pi'_m),(\alpha,\beta)\big),(\alpha,\beta),0,0,\ldots\big) \formspace.
\end{equation}
We can extend each $L_m$ in a unique way to a homeomorphism of the completed BU-algebra $\overline{\BU\big(\Dzs\big)}$, using Schwartz' Kernels Theorem. We denote the extended operator by the same symbol $L_m$. The action of $L_m$ on a homogeneous element $\psi^{(N)}$ of degree $N$, i.\,e. on a section $\psi^{(N)}\in\Gamma_0\big((\TsS\oplus\TsS)^{\boxtimes N}\big)$, can be written out explicitly as a sum of terms of degrees $\le N$. Because $\phi_m,\pi_m,\phi_m'$ and $\pi_m'$ depend continuously on $m>0$, so does the section $(\phi_m-\phi'_m)\oplus(\pi_m-\pi'_m)$ and also the sections $\big((\phi_m-\phi'_m)\oplus(\pi_m-\pi'_m)\big)^{\boxtimes n}$ for each\footnote{This may be shown by induction over $n\ge1$, e.\,g. using the joint continuity Lemma  \ref{lem:jointcontinuity} and noting that the linear map $\gamma\mapsto ((\phi_m-\phi'_m)\oplus(\pi_m-\pi'_m))\boxtimes \gamma$ is weakly continuous in $m>0$ for any section $\gamma$ of any vector bundle.} $n\ge 1$. It follows that the components of $L_m\psi^{(N)}$ also depend continuously on $m>0$. Thus we see that $L_m$ is weakly continuous in $m>0$.

Note that $L_m$ preserves the ideal $\ICCRS$ (just as in the proof of Theorem \ref{thm:field-algebra-homeomorphy}), and hence it also preserves the closed ideal $\overline{\ICCRS}$. The $L_m$ therefore descend to homeomorphisms $\tilde{L}_m$ of the quotient algebra ${\Quotientscale{\overline{\BU\big(\Dzs\big)}}{\ \overline{\ICCRS}}}$, and the weak continuity of $L_m$ in $m>0$ implies the weak continuity of $\tilde{L}_m$ in $m>0$. We note that $\tilde{L}_m$ preserves the dense subalgebra ${\Quotientscale{\BU\big(\Dzs\big)}{\ICCRS}}$ and one may verify directly from Theorem \ref{thm:solution_proca_unconstrained} and the definitions of the relevant maps that $\tilde{L}_m$ acts on this subalgebra as
\begin{equation}
\Lambda_m \comp  \Psi_{\varphi'_{m,j}}^{-1} \comp  \Psi_{\varphi'_{m,j}} \comp \Lambda_m ^{-1}.
\end{equation}
If $\Lambda_m \comp  \Psi_{\varphi_{m,j}}^{-1} (O_m)$ depends continuously on $m>0$, then so does
\begin{equation}
\Lambda_m \comp  \Psi_{\varphi'_{m,j}}^{-1} (O_m) = \tilde{L}_m \comp \Lambda_m \comp  \Psi_{\varphi_{m,j}}^{-1} (O_m)
\end{equation}
by the joint continuity Lemma \ref{lem:jointcontinuity}.
\end{proof}
\subsection{The classical case}\label{sec:zero-mass-limit-classical}
For fixed initial data $\Az,\Ad \in \Omega^1(\Sigma)$ on a fixed Cauchy surface $\Sigma$ there is a family of solutions $A_{m,j}$ to the Proca equation of mass $m>0$ with source term $j\in\Omega^1_0(\M)$. We have seen in Theorem \ref{thm:solution_proca_unconstrained} that these solutions take the form
\begin{equation}\label{eqn:recallsoln}
\langle A_{m,j} , F \rangle_\M = \sum\limits_\pm \langle j , G_m^\mp F   \rangle_{J^\pm(\Sigma)} +\langle \Az , \rhod G_m F \rangle_\Sigma
		- \langle \Ad , \rhoz G_m F \rangle_\Sigma
\end{equation}
for any fixed $F \in \Omega^1_0(\M)$.

We may think of $F$ as the mathematical representation of an experimental setup which measures the field configuration $A$ through the pairing $\langle A,F\rangle_{\M}$ and we wish to investigate for which $F$, if any, we can take the limit $m\to0$ in Equation (\ref{eqn:recallsoln}) above for all choices of $\Sigma$ and all initial data $\Az,\Ad$.
\begin{lemma}\label{lem:limit_existence_classical_equivalence}
For fixed $F \in \Omega^1_0(\M)$, the limit $m\to 0$ of the right-hand side of Equation (\ref{eqn:recallsoln}) exists for all smooth space-like Cauchy surfaces $\Sigma$ and all initial data $\Az,\Ad\in \Omega^1(\Sigma)$, if and only if $F=F'+F''$ with $F'\in\Omega^1_{0,\delta}(\M)$ and
$F''\in\Omega^1_{0,d}(\M)$ such that $\langle j,F''\rangle_{\M}=0$.
\end{lemma}	
\begin{proof}
Suppose that for a given $F \in \Omega^1_0(\M)$ the right-hand side of Equation (\ref{eqn:recallsoln}) converges as $m\to 0$ for all smooth space-like Cauchy surfaces $\Sigma$ and all $\Az,\Ad\in \Omega^1(\Sigma)$. Because we can vary the initial data arbitrarily and independently, all three terms in Equation (\ref{eqn:recallsoln}) must converge separately. In particular, $\lim_{m\to0}\rhoz G_mF$ must exist in a distributional sense.
Recall that $G_mF=m^{-2}E_md\delta F+E_mF$, where the second term is in $\Omega^1(\M)$ and depends continuously on $m\in\mathbb{R}$ by Corollary \ref{Cor:Emcont}. It then follows from the same corollary and from the continuity and linearity of $\rhoz$ that
    \begin{align}
    \rhoz E_0d\delta F &= \lim_{m\to0} \rhoz E_md\delta F\notag\\
    &= \lim_{m\to0} m^2 \rhoz \left(G_mF-E_mF\right)\notag\\
    &= \lim_{m\to0} m^2 \cdot \left(\lim_{m\to0}\rhoz G_mF - \rhoz E_0F\right)=0,
    \end{align}
where we used the existence of the limit of $\rhoz G_mF$. Because this holds on every Cauchy surface, the one-form $E_0d\delta F$ must annihilate every space-like vector at every point. Because all tangent vectors are linear combinations of space-like vectors we conclude that $E_0d\delta F=0$ and hence also $E_0\delta dF=E_0(\delta d+d\delta)F=0$. We may then define $F' \coloneqq E_0^+\delta dF=E_0^-\delta dF$ and 
$F'' \coloneqq E_0^+d\delta F=E_0^-d\delta F$ and note that these have compact supports. Furthermore, since $\delta$ and $d$ intertwine with $E_0^+$ on forms,  $\delta F'=0=dF''$ and
	\begin{equation}
	F' + F '' = E_0^+ (d \delta + \delta d) F = F \formspace.
	\end{equation}
Combining this formula with $G_m^\pm=E_m^\pm(m^{-2}d\delta+1)$ we find
    \begin{align}\label{eqn:GmF}
    G_m^\pm F &=E_m^\pm F' + m^{-2}E_m^\pm (d\delta+\delta d+m^2)F''\notag\\
    &=E_m^\pm F' + m^{-2}F''.
    \end{align}
Substituting this in the first term of Equation (\ref{eqn:recallsoln}) we see that
    \begin{equation}\label{eqn:jterm}
    \sum\limits_\pm \langle j , G_m^\mp F   \rangle_{J^\pm(\Sigma)}
     = \sum\limits_\pm \langle j , E_m^\pm F' \rangle_{J^\pm(\Sigma)} + m^{-2} \langle j , F'' \rangle_{\M}
    \end{equation}
must converge as $m\to0$. The terms in the first sum converge as $m\to0$ by Corollary \ref{Cor:Emcont}, and hence the last term must also converge. This clearly implies $\langle j , F'' \rangle_{\M}=0$, showing that $F$ must have the stated form.

Conversely, when $F=F'+F''$ with $\delta F'=0=dF''$ and $\langle j , F'' \rangle_{\M}=0$, then it follows from Equation (\ref{eqn:GmF}) that $G_mF=E_mF'$, which has a limit as $m\to0$. Together with Equation (\ref{eqn:jterm}) and the continuity of $\rhod$ and $\rhoz$ it follows that the right-side of Equation (\ref{eqn:recallsoln}) converges as $m\to0$.
\end{proof}
Note that $F'$ and $F''$ are uniquely determined by $F=F'+F''$ and $\delta F'=0=dF$, because $\Omega^1_{0,\delta}(\M){\hspace{0.1em}\cap\hspace{0.1em}}\Omega^1_{0,d}(\M)=\{0\}$. Indeed, if $\tilde{F}\in \Omega^1_0(\M)$ satisfies $d\tilde{F}=\delta\tilde{F}=0$, then also $\square \tilde{F}=0$ and hence $\tilde{F}=0$ by \cite[Corollary 3.2.4]{BaerGinouxPfaeffle2007}.\par 
For a fixed $m>0$ and $j\in\Omega^1(\M)$ there are $F\in\Omega^1_0(\M)$ which define trivial observables in the sense that
$\langle A_{m,j},F\rangle_{\M}=0$ for all field configurations (i.\,e. for all initial data). The following lemma characterizes them:
\begin{lemma}\label{lem:trivialobservables}
For fixed $m>0$ and $j\in\Omega^1(\M)$, $F\in\Omega^1_0(\M)$ defines a trivial observable if and only if $F=(\delta d+m^2) \tilde{F}$ for some $\tilde{F}\in\Omega^1_0(\M)$ with $\langle j,\tilde{F}\rangle_{\M}=0$.
\end{lemma}
\begin{proof}
Arguing as in the proof of Lemma \ref{lem:limit_existence_classical_equivalence} we see that $F$ defines a trivial observable if and only if $G_mF=0$ and $\langle j,G_m^+F\rangle_{\M}=0$. The first condition is equivalent to $F=(\delta d+m^2)\tilde{F}$ for some $\tilde{F}\in\Omega^1_0(\M)$ by Lemma \ref{lem:dynamics}. The second condition then means that $\langle j,\tilde{F}\rangle_{\M}=0$.
\end{proof}
For any fixed $m$ and $j$ one would normally divide out these trivial observables, because they are redundant. For our purposes, however, this is rather awkward, because the space of trivial observables depends on $m$ and $j$. However, we can remove some of the redundancy in the following way:
\begin{theorem}[Existence of the zero mass limit]\label{thm:limit_existence_classical}
Fix $j\in\Omega^1(\M)$. For $F\in\Omega^1_0(\M)$, Equation (\ref{eqn:recallsoln}) admits a massless limit for all initial data on all Cauchy surfaces if and only if
there is a $F'\in\Omega^1_{0,\delta}(\M)$ such that $F-F'$ is a trivial observable for all $m>0$.
\end{theorem}
\begin{proof}
It follows from Lemma \ref{lem:limit_existence_classical_equivalence} that $F=F'+F''$ with $F' \in \Omega^1_{0,\delta}(\M)$ and $F''\in\Omega^1_{0,d}(\M)$ such that $\langle j, F''\rangle_{\M}=0$. From Equation (\ref{eqn:GmF}) we see that for all $m>0$ it holds $G_mF''=0$ and $\langle j,G_m^+F''\rangle_{\M}=m^{-2}\langle j,F''\rangle_{\M}=0$, so $F''=F-F'$ defines a trivial observable for all $m>0$ by Lemma \ref{lem:trivialobservables} and its proof.
\end{proof}
%
In other words, for the massless limit it suffices\footnote{It is unclear if there is any remaining redundancy.}  to consider all co-closed forms $\Omega^1_{0,\delta}(\M)$. The meaning of this can be quite easily understood under the duality $\langle {\PH} , {\PH} \rangle_\M$. One finds that $\Quotientscale{\mathcal{D}^1(\M)}{d\mathcal{D}^{0}(\M)}$ is dual to $\Omega^1_{0,\delta}(\M)$ (see \cite[Section 3.1]{SandersDappiaggiHack2014}). Here, $\mathcal{D}^1(\M)$ denotes the set of distributional one-forms (in a physical sense, these are classical vector potentials), so restricting to co-closed test one-forms is equivalent to implementing the gauge equivalence $A \to A + d\chi$, for $A \in \mathcal{D}^1(\M)$ and $ \chi \in \mathcal{D}^0(\M)$ in the theory. This dual relation is easily checked for $A' = A + d\chi$ and $F \in \Omega^1_{0,\delta}(\M)$
\begin{align}
\langle A', F \rangle_{\M}
&= \langle A, F \rangle_{\M} + \langle d\chi, F \rangle_{\M} \notag\\
&= \langle A, F \rangle_{\M} + \langle \chi, \delta F \rangle_{\M} = \langle A, F \rangle_{\M} \formspace.
\end{align}
This is a nice result, because it elucidates the gauge equivalence in the Maxwell theory. Note that it is a priori unclear how to implement the gauge equivalence in Maxwell's theory on curved spacetimes due to the non-trivial topology. Maxwell's equation $\delta d A = 0$ suggests that two solutions that differ by a closed one-form give rise to the same configuration, but one can argue that only exact one-forms should be treated as pure gauge solutions, because the Aharonov-Bohm effect \emph{does} distinguish between configurations that differ by a form that is closed but not exact \cite{SandersDappiaggiHack2014}. It is gratifying to see that we arrive at a gauge equivalence given by the class of exact forms, simply by keeping the set of linear observables as large as possible in the limit, i.\,e.~$\Omega^1_{0,\delta}(\M)$.

Hence, we have already captured one important feature of the Maxwell theory in the massless limit of the Proca theory! It remains to check whether also the dynamics are well behaved in the massless limit.
\subsubsection{Dynamics and the zero mass limit}\label{sec:limit_dynamics_classical}

In the massless limit one may hope to find a vector potential $A_{0,j}$ satisfying Maxwell's equations $\delta dA_{0,j}=j$ at least in a distributional sense, i.\,e.
$\langle A_{0,j} , \delta d F \rangle_\M=\langle j , \delta d F \rangle_\M$ for every test one-form $F\in\Omega^1_0(\M)$. Note that $\delta dF$ is co-closed, so by Theorem \ref{thm:limit_existence_classical} we may substitute $\tilde{F}=\delta dF$ in the limit 
\begin{align}
\langle A_{0,j} , \tilde{F} \rangle_\M
&\coloneqq  \lim\limits_{m \to 0} \langle A_{m,j}, \tilde{F} \rangle_\M \\
&=  \lim\limits_{m \to 0}\Big( \sum\limits_\pm \langle j , G_m^\mp \tilde{F}   \rangle_{J^{\pm}(\Sigma)} +\langle \Az , \rhod G_m \tilde{F} \rangle_\Sigma
- \langle \Ad , \rhoz G_m \tilde{F} \rangle_\Sigma \Big) \formspace\notag
\end{align}
for any given initial data $\Az,\Ad$ on any Cauchy surface $\Sigma$. However, using
\begin{align}
\lim_{m\to0}G_m^{\pm}\delta dF
&=\lim_{m\to0}E_m^{\pm}\delta dF= E_0^{\pm}\delta dF\notag\\
&=F - E_0^{\pm}d\delta F
\end{align}
we only find
\begin{align}\label{eqn:dynamics_limit_classical_unconstraint}
\langle A_{0,j} , \delta d F \rangle_\M&=
 \sum\limits_\pm \langle j , F-E_0^{\mp}d \delta F\rangle_{J^\pm(\Sigma)} - \langle \Az , \rhod E_0 d\delta F \rangle_\Sigma
+ \langle \Ad , \rhoz E_0 d\delta F \rangle_\Sigma\notag\\
&=
\langle j , F\rangle_{\M}
- \sum\limits_\pm \langle j , E_0^{\mp}d \delta F\rangle_{J^\pm(\Sigma)}
+ \langle \Ad , \rhoz E_0 d\delta F \rangle_\Sigma,\formspace
\end{align}
where we used the fact that $\rhod E_m d\delta  F = - *_{(\Sigma)}i^* * d E_m d \delta F = 0$ since $d$ and $E_m$ commute.
The second term in Equation (\ref{eqn:dynamics_limit_classical_unconstraint}) will not vanish in general (e.\,g. when $dF=0$ but $\langle j,F\rangle_{\M}\not=0$). Ergo, the fields $A_{0,j}$ defined as the zero mass limit of the Proca field $A_{m,j}$ will not fulfill Maxwell's equation in a distributional sense. While this might seem surprising at first, it is quite easy to understand when we recall how we have found solutions to Proca's equation, using the massive wave equation (\ref{eqn:classical_wave_eqation}) combined with constraint equations on the initial data to ensure that the Lorenz constraint (\ref{eqn:classical_constraint}) is fulfilled. Similarly, one solves Maxwell's equation by specifying a solution to the massless wave equation $(\delta d + d \delta )A_{0,j} = j$ and restricting the initial data such that the Lorenz constraint $\delta A_{0,j} = 0$ is fulfilled. The problem in the massless limit lies with the constraints. Recall from Theorem \ref{thm:solution_proca_unconstrained} that, in order to implement the Lorenz constraint, we have restricted the initial data by
\begin{equation}
\Adelta = m^{-2}\rhodelta j \formspace , \quad \textrm{and}  \quad
\An = m^{-2}\left( \rhon j  + \delta_{(\Sigma)} \Ad \right) \formspace.
\end{equation}
It is obvious that, in general, the resulting $\Adelta$ and $\An$ diverge in the zero mass limit, so there is no corresponding solution to Maxwell's equations with the same initial data. In order to keep the dynamics in the zero mass limit, we need to make sure that the constraints are well behaved in the limit. Since we do not want the external source or the initial data to be dependent of the mass, we have to require that $\Adelta$ and $\An$ vanish, i.\,e. we need to specify\footnote{The first equation follows from $\rhodelta j=0$ on all Cauchy surfaces.}
\begin{align}
\delta j &= 0 \formspace , \quad \textrm{and} \\
\delta_{(\Sigma)} \Ad &= -\rhon j \formspace.
\end{align}
 This corresponds exactly to the constraints on the initial data for the Maxwell equation which implement the Lorenz gauge (cf. Pfenning \cite[Theorem 2.11]{Pfenning2009}). With these constraints, we can now look at the remaining term of $\langle A_{0,j} , \delta d F\rangle_\M$ in Equation (\ref{eqn:dynamics_limit_classical_unconstraint}). We do this separately for the two summands. Using that $d$ commutes with pullbacks and inserting the constraints on the initial data, we find
\begin{align}
\langle \Ad , \rhoz E_0 d\delta  F \rangle_\Sigma
&= \langle \Ad , d_{(\Sigma)} \rhoz E_0 \delta  F \rangle_\Sigma \notag\\
&=	\langle \delta_{(\Sigma)}\Ad ,  \rhoz E_0 \delta  F \rangle_\Sigma \notag\\
&= -\langle \rhon j ,  \rhoz E_0 \delta  F \rangle_\Sigma 
\end{align}
For the first summand $\sum_\pm \langle j ,  E_0^\mp d\delta F   \rangle_{\Sigma^\pm}$ we use the partial integration in Equation (\ref{eqn:partialintegration}) in the proof of Theorem \ref{thm:solution_proca_unconstrained} and find, using $m=0$ and the constraint $\delta j = 0$ as specified above,
\begin{align}
\sum_\pm \langle j ,  E_0^\mp d\delta F   \rangle_{J^\pm(\Sigma)}
&= \sum_\pm \langle d \delta j ,  E_0^\mp  F   \rangle_{J^\pm(\Sigma)} + 
\langle \rhodelta E_0F,\rhon j\rangle_\Sigma - \langle \rhodelta j,\rhon E_0F\rangle_\Sigma\notag\\
&= \langle \rhoz E_0\delta F,\rhon j\rangle_\Sigma  \formspace.
\end{align}
Using the symmetry of the inner product $\langle {\PH},{\PH}\rangle_{\M}$ we find that the remaining terms of Equation (\ref{eqn:dynamics_limit_classical_unconstraint}) cancel when restricting the initial data such that they are well defined in the zero mass limit. We therefore obtain the correct dynamics in that case:
\begin{align}
\langle A_{0,j} , \delta d F \rangle_\M
&=  \langle j , F \rangle_\M - \lim\limits_{m \to 0}\Big(\sum\limits_\pm \langle j ,  E_m^\mp d\delta F   \rangle_{J^\pm(\Sigma)}
- \langle \Ad , \rhoz E_m d\delta  F \rangle_\Sigma \Big) \notag\\
&= \langle j , F \rangle_\M \formspace.
\end{align}
In combination with Theorem \ref{thm:limit_existence_classical} we have thus shown
\begin{theorem}[The zero mass limit of the Proca field]
	Let $F\in \Omega^1_0(\M)$ be a test one-form and $j \in \Omega^1(\M)$ an external current.
	Let $A_{m,j}$ be the solution to Proca's equation specified by initial data $\Az, \Ad \in \Omega^1_0(\Sigma)$ via Theorem \ref{thm:solution_proca_unconstrained}.	\par
	Defining the zero mass limit $\langle A_{0,j} , F \rangle_\M = \lim_{m \to 0} \langle A_{m,j}, F \rangle_\M$ of the Proca field, the following holds:
	\begin{enumerate}
		\item The limit exists if and only if $F$ is equivalent to an observable $F'$ (for all $m>0$) with $\delta F' = 0$, effectively implementing the gauge equivalence of the Maxwell theory.
		\item The field $A_{0,j}$ is a Maxwell field, that is, it solves Maxwell's equation, if and only if the current is conserved, $\delta j = 0$, and $\rhon j = - \delta_{(\Sigma)} \Ad$, implementing the Lorenz gauge.
	\end{enumerate}
\end{theorem}
Note that the conservation of the external current $\delta j = 0$ is not required to solve Proca's equation, but it is necessary to solve Maxwell's equations ($\delta dA=j$ entails $\delta j=0$). It is therefore not surprising that this condition is also necessary to recover the dynamics in the zero mass limit.
In analogy to the quantum theory, we may think of the field configuration $A$ as a state, whereas $F$ is an observable. We then see from the theorem that the limits of observables give rise to the gauge equivalence of the classical vector potential, but additional conditions on the limits of states and external currents are needed in order to recover Maxwell's equation.

\subsection{The quantum case}\label{sec:mass_dependence_and_limit}
In the quantum case we define the observables in the zero mass limit as follows:
\begin{definition}[Zero mass limit theory]\label{def:limittheory}
	For any fixed $j \in \Omega^1(\M)$ and $O\in \BUOmega$ we say that $[O]_{m,j}\in \AA_{m,j}$ has a zero mass limit if and only if
	\begin{equation}
	\lim_{m\to0}\big( \Lambda_m \comp  \Psi_{\varphi_{m,j}}^{-1} \big) \big([O]_{m,j}\big)\notag
	\end{equation}
	exists for all Cauchy surfaces $\Sigma$ and all families $\left\{ \varphi_{m,j} \right\}_{m\ge0}$ of classical solutions to the inhomogeneous equation $(\delta d + m^2) \varphi_{m,j} = j$ which depend continuously on $m$. Here, $\Lambda_{m}$ and $\Psi_{\varphi_{m,j}}$ are as defined in Section \ref{sec:CCR} and \ref{sec:no-current}.\par
	We call the zero mass limit trivial if and only if the above limit vanishes for all Cauchy surfaces $\Sigma$ and all families $\{ \varphi_{m,j}\}_{m\ge0}$. If the zero mass limit exists, we denote its equivalence class modulo trivial observables by $[O]_{0,j}$.
\end{definition}
Note that we included $m=0$ in the family $\left\{ \varphi_{m,j} \right\}_{m\ge0}$. This is done for the following reason. Even when $j=0$ we may choose a non-trivial family $\left\{ \varphi_{m,0} \right\}_{m\ge0}$ and due to the isomorphism $\Psi_{\varphi_{m,0}}^{-1}$ we are then considering quantum fluctuations around the classical solutions $\varphi_{m,0}$. If the quantum field is to converge, it seems reasonable to require that the classical background field $\varphi_{m,0}$ also converges. For general sources this implies that $\varphi_{0,j}$ satisfies Maxwell's equations and hence the current must be conserved, $\delta j=0$.
%
%

We can think of the zero mass limit of an operator $O$ as a family of operators in the algebras ${\Quotientscale{\BU\big( \Dzs \big) }{\ICCR}}$, indexed by $\Sigma$ and by the family $\left\{ \varphi_{m,j} \right\}_{m\ge0}$. Using the properties of the topological algebras ${\Quotientscale{\BU\big( \Dzs \big) }{\ICCR}}$, it is not hard to see that the operators $O\in \BUOmega$ which have a zero mass limit form a $^*$-subalgebra of $\BUOmega$ in which the operators with a trivial zero mass limit form an ideal. We are interested in the quotient algebra which we denote by $\AA_{0,j}$ and which is generated by $\mathbbm{1}$ and by homogeneous degree-one elements, which we denote by $\A_{0,j}(F)$. These are the massless field operators and we can think of them as the massless limits of the field operators $\A_{m,j}(F)$. Our next theorem focuses on these field operators.\par 
As our main result we determine for which $F\in\Omega^1_0(\M)$ the limit $\A_{0,j}(F)$ exists.
\begin{theorem}[Existence of the zero mass limit]\label{thm:limit_existence_sourcefree}
For given $j\in\Omega^1_{\delta}(\M)$, $\A_{m,j}(F)$ has a zero mass limit $\A_{0,j}(F)$ if and only if $F\in\Omega^1_0(\M)$ is of the form $F=F'+F''$ with $F'\in\Omega^1_{0,\delta}(\M)$ and $F''\in\Omega^1_{0,d}(\M)$ such that $\langle j,F''\rangle_{\M}=0$. The zero mass limit is trivial when $F'=0$.
\end{theorem}
\begin{proof}
Note that
\begin{equation}
\big(\Lambda_m \comp \Psi^{-1}_{\varphi_{m,j}}\big)\big(\A_{m,j}(F)\big)
 =\big[\big(\langle \varphi_{m,j},F\rangle_{\M},\kappa_mF,0,0,\ldots\big)\big]_\sim^{\text{CCR}}\formspace .
\end{equation}
Just as in the last paragraph of the proof of Lemma \ref{lem:limit_existence_classical_equivalence} we see that all $F$ of the stated form have a limit $\lim_{m\to0}G_mF=\lim_{m\to0}E_mF$ and hence the limit of the initial data $\lim_{m\to0}\kappa_mF$ exists on every Cauchy surface. By assumption on the $\varphi_{m,j}$, 
$\langle \varphi_{m,j},F\rangle_{\M}$ also has a limit as $m\to0$. Because $[{\PH}]_\sim^\text{CCR}$ is continuous and independent of $m$ we see that $\lim_{m\to0}
\big(\Lambda_m \comp \Psi^{-1}_{\varphi_{m,j}}\big)\big(\A_{m,j}(F)\big)$ exists for all $F$ of the stated form.\par 
When $F'=0$, then $F=F''$ and $G_mF=0$ (cf. the proof of Theorem \ref{thm:limit_existence_classical}) and hence $\kappa_mF$ on every Cauchy surface. Furthermore, $\langle\varphi_{m,j},F\rangle_{\M}=m^{-2}\langle j,F''\rangle_{\M}=0$ by Theorem \ref{thm:solution_proca_unconstrained} and Equation (\ref{eqn:GmF}). Thus the zero mass limit is trivial.\par
Assume that $\lim_{m\to0}\big(\Lambda_m \comp \Psi^{-1}_{\varphi_{m,j}}\big)\big(\A_{m,j}(F)\big)$ exists. This means that for each Cauchy surface $\Sigma$ there is a family of elements $g_m\in\ICCR$ such that $\lim_{m\to0} \big(\langle \varphi_{m,j},F\rangle_{\M},\kappa_mF,0,\ldots\big)+g_m$ exists in $\BU\left(\Dzs \right)$. Using the projection $S$ of Lemma \ref{lem:symmetrization-of-fields}, we have
\begin{align}
S\big(\big(\langle \varphi_{m,j},F\rangle_{\M},\kappa_mF,0,0, \ldots\big)+g_m\big)
&=\big(\langle \varphi_{m,j},F\rangle_{\M},\kappa_mF,0,0, \ldots\big)\notag\\
&=S\big(\langle \varphi_{m,j},F\rangle_{\M},\kappa_mF,0,0, \ldots\big)\formspace,
\end{align}
because $\big(\langle \varphi_{m,j},F\rangle_{\M},\kappa_mF,0,\ldots\big)$ is homogeneous of degree 1 and hence symmetric. The continuity of $S$ then implies that
\begin{align}
S\big(\lim_{m\to0} \big(\langle \varphi_{m,j},F\rangle_{\M},\kappa_mF,0,0, \ldots\big)+g_m\big)
&=\lim_{m\to0} S\big(\big(\langle \varphi_{m,j},F\rangle_{\M},\kappa_mF,0,0, \ldots\big)+g_m\big)\notag\\
&=\lim_{m\to0} \big(\langle \varphi_{m,j},F\rangle_{\M},\kappa_mF,0,0, \ldots\big),
\end{align}
exists. This implies that both $\lim_{m\to0}\langle \varphi_{m,j},F\rangle_{\M}$ and $\lim_{m\to0}\kappa_mF$ exist. The first of these conditions already follows from the assumptions on $\varphi_{m,j}$ but the second implies in particular that $\lim_{m\to0}\rhoz G_mF$ exists. Because this is required for every Cauchy surface, the argument presented in the proof of Lemma \ref{lem:limit_existence_classical_equivalence} shows that $F$ must be of the stated form.
\end{proof}
As in the classical case we find that the algebra $\AA_{0,j}$ of the massless limit is generated by field operators $\A_{0,j}(F)$ with $F\in\Omega^1_{0,\delta}(\M)$ ranging over the co-closed test one-forms. Just as in the classical case, discussed in Section \ref{sec:zero-mass-limit-classical}, this implements the gauge equivalence of the Maxwell theory, using the choice of gauge equivalence of \cite{SandersDappiaggiHack2014}. Hence also in the quantum case, the limit exists only if we implement the gauge beforehand.
We now turn to the algebraic relations in $\AA_{0,j}$. For this we view $[O]_{0,j}$ as an equivalence class of a family of limits 
$\lim_{m\to0}\big( \Lambda_m \comp  \Psi_{\varphi_{m,j}}^{-1} \big) \big([O]_{m,j}\big)$ in the algebras ${\Quotientscale{\BU\big( \Dzs \big) }{\ICCR}}$, indexed by the Cauchy surface $\Sigma$ and the family $\left\{ \varphi_{m,j} \right\}_{m\ge0}$ and we set in particular $\A_{0,j}(F) \coloneqq [(0,F,0,\ldots)]_{0,j}$. Exploiting the algebraic structure of the algebras ${\Quotientscale{\BU\big( \Dzs \big) }{\ICCR}}$ we then find in a natural way\footnote{This means that the relations below hold for the corresponding limits $\lim_{m\to0}\big( \Lambda_m \comp  \Psi_{\varphi_{m,j}}^{-1} \big) ([O]_{m,j})$ for each Cauchy surface and for each family of classical solutions $\left\{ \varphi_{m,j} \right\}_{m\ge0}$.} that
\begin{align}
\A_{0,j}(\alpha F + \beta F') &= \alpha \A_{0,j}(F) + \beta \A_{0,j}(F')  \\
\A_{0,j}(F)^* &= \A_{0,j}(\skoverline{F}\,) 	
\end{align}
for all $F \in \Omega^1_{0,\delta}(\M)$ and $\alpha, \beta \in \IC$, corresponding to the linearity and the hermitian field property. For the canonical commutation relations we note that for all $F,F' \in \Omega^1_{0,\delta}(\M)$, $G_mF'=E_mF'$ and hence
\begin{align}\label{eqn:limCCR}
\big[ \A_{0,j}(F) ,\A_{0,j}(F')  \big]
&= \lim\limits_{m \to 0} \big[ \A_{m,j}(F) ,\A_{m,j}(F')  \big]  \notag\\
&=\i \cdot \lim\limits_{m \to 0} \Gm{F}{F'}\cdot \mathbbm{1} \notag\\
&=\i \cdot \lim\limits_{m \to 0} \langle F,E_F'\rangle_{\M}\notag\\
&= \i \, \Ez{F}{F'}\cdot \mathbbm{1}\formspace.
\end{align}
For co-closed test one-forms $F \in \Omega^1_{0,\delta}$, the fundamental solutions $E^\pm_0$ of the massless Klein-Gordon operator are actually also fundamental solutions to Maxwell's equation, i.\,e. it holds $E_0^\pm \delta d F = E_0^\pm (\delta d + d \delta) F = F$, so we find that the fields in the zero mass limit are subject to the correct canonical commutation relations. Indeed, using $\rhodelta E_0 F' = i^* \delta E_0 F' = i^* E_0 \delta F' = 0$ and the analogous expression for $F$, we may rewrite commutator in terms of initial data as
\begin{align}
\Ez{F}{F'}
&=\langle F, E_0 F' \rangle_\M = -\langle E_0 F , F' \rangle_\M \notag\\
&= \langle \rhoz E_0 F , \rhod E_0 F'  \rangle_\Sigma - \langle \rhod E_0 F , \rhoz E_0 F' \rangle_\Sigma
\end{align}
in analogy to Equation (\ref{eqn:CCRdata}).

Note that $\Ez{F}{F'}$ for $F,F' \in \Omega^1_{0,\delta}(\M)$ is in general degenerate, hence the quantum field theory associated with $\A_{0,j}$ will in general fail to be local in the sense of Definition \ref{def:generally-coveriant-qftcs}. However, this is perfectly in line with the free vector potential as presented in \cite{SandersDappiaggiHack2014}.\par

It remains to verify whether $\A_{0,j}$ solves Maxwell's equation, i.\,e. if $\A_{0,j}(\delta d F) = \langle j , F\rangle_\M$ holds for all $F\in \Omega^1_0(\M)$. Because $\delta d F$ is co-closed, the limit $\A_{0,j}(\delta dF)$ is well defined. For any Cauchy surface and any family $\left\{ \varphi_{m,j} \right\}_{m\ge0}$ we have
\begin{align}\label{eqn:Maxwellfail}
\big(\Lambda_m \comp \Psi^{-1}_{\varphi_{m,j}}\big)\big(\A_{m,j}(\delta dF)\big)
&=\big[\big(\langle \varphi_{m,j},\delta dF\rangle_{\M},\kappa_m\delta dF,0,0, \ldots\big)\big]_\sim^{\text{CCR}}\notag\\
&=\big[\big(\langle \delta d\varphi_{m,j},F\rangle_{\M},\kappa_m\delta dF,0,0, \ldots\big)\big]_\sim^{\text{CCR}}\notag\\
&=\langle j,F\rangle_{\M} \mathbbm{1} + \big[\big(0,\kappa_m\delta dF,0,0, \ldots\big)\big]_\sim^{\text{CCR}},
\end{align}
which is independent of $\left\{ \varphi_{m,j} \right\}_{m\ge0}$. This essentially means that it suffices to consider the source free case, because the second term in Equation (\ref{eqn:Maxwellfail}) is $\Lambda_m\big(\A_{m,0}(\delta dF)\big)$. Because $G_m\delta dF=E_m\delta dF$ converges to $E_0\delta dF$ we have
\begin{equation}
\lim_{m\to0}\kappa_m\delta dF=\big(\rhoz E_0\delta dF,\rhod E_0\delta dF\big)=\big(\rhoz E_0\delta dF,0\big),
\end{equation}
where we have used that $E_0\delta dF=-E_0d\delta F$ is closed and hence $\rhod E_0\delta dF=\rhon dE_0\delta dF=0$.\par
To recover Maxwell's equation, we need to verify that the second term in Equation (\ref{eqn:Maxwellfail}) vanishes in the limit $m\to0$ for any Cauchy surface. However, this fails in general. Indeed, if $B\in\Omega^1(\M)$ is the solution of the wave equation $\Box B=0$ with initial data $\rhoz B=\rhod B=\rhon B=0$ and $\rhodelta B\in\Omega^0_0(\M)$ not constant, then $B=E_0F$ for some compactly supported $F\in\Omega^1_0(\M)$ (cf. the proof of Lemma \ref{lem:one-particle-homeomorphism}). However, $E_0\delta dF=-E_0d\delta F=-d\delta E_0F=-d\delta B$ does not vanish, because $\delta B\in\Omega^0(\M)$ is a function which is not constant. In particular, because $d$ commutes with pull-backs, $\rhoz E_0\delta dF=-d_{\Sigma}\rhodelta B\not\equiv0$ because $\rhodelta B$ is not a constant function. Conversely, following the proof of Theorem \ref{thm:limit_existence_sourcefree} and Lemma \ref{lem:limit_existence_classical_equivalence} we see that the limit only vanishes for all Cauchy surfaces if $E_0\delta dF=0$, which means that $F\in \Omega^1_{0,\delta}(\M)+\Omega^1_{0,d}(\M)$.\par 
We have encountered a similar situation in the investigation of the classical theory in Section \ref{sec:limit_dynamics_classical} (cf. Equation (\ref{eqn:dynamics_limit_classical_unconstraint})). There we could get rid of similar remaining terms by restricting the initial data of the field configuration (i.\,e. of the state of the system) such that the Lorenz constraint is well behaved in the limit. In the quantum scenario, our definition of the massless limit already requires $\delta j=0$, but the remaining constraint equation has not been imposed. Indeed, in our present setting, which focuses on observables, the Lorenz constraint does not appear directly at all.\par 
Nevertheless, we may impose the desired dynamics in a consistent way by dividing out a corresponding ideal. Note in particular that the limit algebra is not simple, because the skew-symmetric form in Equation (\ref{eqn:limCCR}) is degenerate: $\langle F,E_0\delta dF'\rangle_{\M}=0$ when $\delta F=0$. It follows that the operators $\A_{0,j}(\delta dF)-\langle j,F\rangle_{\M}\mathbbm{1}$ commute with all other operators in the algebra $\AA_{0,0}$ and they therefore generate a two-sided ideal.\par 
In the source free case this ideal is generated by the operators $\A_{0,j}(\delta dF)$, which correspond to 
$[(0,\kappa_m\delta dF,0,\ldots)]_\sim^{\text{CCR}}$ with $\kappa_m\delta dF=(\rhoz E_0\delta dF,0)$. It is interesting to note that 
$A_F \coloneqq E_0 \delta dF$ is a space-like compact solution to the source free Maxwell equation, $\delta d A_F = -\delta d  E_0 d \delta F = 0$, and that it is of the form 
$A_F=d\chi$ with the space-like compact function $\chi \coloneqq -E_0\delta F$. Solutions of the form $A_F$ can also be characterized in terms of their initial data,
\begin{equation}
\big(\rhoz A_F, \rhod A_F\big) =  \big(-d_{(\Sigma)} \rhoz \chi, 0\big) \formspace.
\end{equation}
Under the correspondence $F\mapsto E_0F$ of observables (with $\delta F=0$) and space-like compact solutions to Maxwell's equation, the observables $\delta dF$ therefore generate a subspace that looks like a kind of pure gauge solutions (see for example \cite{SandersDappiaggiHack2014} or \cite{Pfenning2009}). However, the kind of ``gauge equivalence'' on the level of the observables, rather than the fields, does not seem to come out of the limiting procedure naturally.

It seems plausible that one can recover the correct dynamics by including states in the investigation and formulating conditions on their limiting behaviour, which essentially require that the remaining constraint equations is well behaved in the limit. It is unclear if our limiting procedure can also be improved to directly recover the dynamics without considering states. One idea is to consider the homeomorphisms that propagate the algebras of initial data $\Quotientscale{\BU\big(\Dzs\big)}{\ICCR}$ from one Cauchy surface to another. If one can formulate a condition that ensures that these homeomorphisms remain well behaved in the limit, then the resulting limits should have a well behaved time evolution. It would be of interest to develop these ideas and to compare the results with the massless limit of Stueckelberg's theory, which preserves the gauge invariance at all masses at the cost of introducing a coupling to an additional scalar field and all the associated additional complications \cite{BelokogneFolacci2016}. We leave the investigation of these worthwhile questions to the future.
%
%
%

%
\section{Conclusion and Outlook} \label{chpt::conclusion}
We have studied the classical and quantum Proca field in curved spacetimes, using a general setting including external sources and without restrictive assumptions on the spacetime topology. We have shown that the quantum theory is locally covariant in the sense of \cite{BrunettiFredenhagenVerch2003}, where the injectivity of the morphisms is related to the non-degeneracy of the symplectic form.\par
We have shown that the theory depends continuously on the mass $m>0$, in a way which we have defined. Using specific BU-algebra homeomorphism we mapped families of smeared Proca fields at different masses, initially elements in different BU-algebras, into the BU-algebra of initial data. The topology of the latter algebra
then determines a notion of continuity for the family of operators. For $m>0$ we showed that this notion of continuity is independent of the choice of Cauchy surface and of the classical inhomogeneous solutions $\varphi_{m,j}$ appearing in the homeomorphisms. This result relied crucially on the use of energy estimates. Note that a $C^*$-Weyl algebra approach is ill-suited for the investigation of the zero mass limit, as one of us has argued in \cite[Appendix A]{Schambach2016}.\par
For the quantum theory we defined the zero mass limit by requiring a continuous family of observables to converge on every Cauchy surface and for every continuous family $\left\{ \varphi_{m,j} \right\}_{m\ge0}$ of inhomogeneous classical solutions. (For the classical theory we considered a somewhat simplified setting.) Investigating the zero mass limit we found in both cases that the limit exists and the theory is generated by the class of observables described by co-closed test one-forms. This effectively implements a gauge invariance on the (distributional) solutions to Proca's equation by exact (distributional) one-forms. This is of interest, because in general curved spacetimes the spacetime topology allows different possible choices of gauge invariance (using e.\,g. closed forms instead). Our limiting procedure naturally leads to the same gauge invariance that was advocated in \cite{SandersDappiaggiHack2014}, using the independent argument that it can account for phenomena such as the Aharonov-Bohm effect and Gauss' law.\par
In the zero mass limit we also find that the quantum fields fulfill the basic properties of linearity, the hermitian field property and the correct CCR, all in line with the massless vector potential of electrodynamics. However, we do not automatically recover the expected Maxwell dynamics. In the classical case, this is caused by a potential divergence in the constraint equations on the initial data of field configurations. This may be avoided by requiring the external source to be conserved, $\delta j=0$, and by requiring that the initial data of the configuration also satisfy the constraint equations of Maxwell's theory as given e.\,g. by Pfenning \cite{Pfenning2009}. In the quantum case we did not clarify if Maxwell's equation can be obtained in the zero mass limit, e.\,g. by imposing additional conditions on the limits of observables or on states, or by requiring the homeomorphisms that propagate initial data between different Cauchy surfaces to remain well defined in the massless limit.\par
The further development of these ideas might require a detailed investigation of Hadamard states, which is also if interest in its own right. So far these states seem to have been considered only in a restricted class of spacetimes \cite{FewsterPfenning2003}. Furthermore, it would be interesting to make a detailed comparison of our massless limit and the massless limit of Stueckelberg's theory as presented e.\,g. in \cite{BelokogneFolacci2016}. We leave the investigation of these worthwhile questions to the future.\par
\textbf{Acknowledgements}
We would like to thank the University of Leipzig, where this research was carried out, and MS would like to thank Prof. Stefan Hollands for helpful comments and discussions. Large  parts of this work are adapted from the MSc thesis of MS.
%
%
%
%
%
%
%
%
%
%


\appendix
\section{Additional Lemmas}\label{app:lemmata}
Let $\mathfrak{X}$ be a complex vector bundle over a smooth differential manifold $\N$. As in Section \ref{sec:BU-algebra} we may define the complete BU-algebra $\overline{\BU(\Gamma_0(\mathfrak{X}))}$ over $\Gamma_0(\mathfrak{X})$ as the direct sum
\begin{align}
\overline{\BU\big(\Gamma_0(\mathfrak{X})\big)} = \IC \oplus \bigoplus\limits_{n= 1}^\infty \Gamma_0\big(\mathfrak{X}^{\boxtimes n}\big) \formspace,
\end{align}
using the outer tensor product of vector bundles (see \cite[Chapter 3.3]{SahlmannVerch2000}). We endow this algebra with the inductive limit topology of the subspaces
\begin{align}
\BU_N = \IC \oplus \bigoplus\limits_{n= 1}^N \Gamma_0(\mathfrak{X}^{\boxtimes n}) \formspace.
\end{align}
Note that $\overline{\BU\big(\Gamma_0(\mathfrak{X})\big)}$ is the completion of the BU-algebra
$\IC \oplus \bigoplus\limits_{n= 1}^\infty \Gamma_0(\mathfrak{X})^{\otimes n}$.
\begin{lemma}\label{lem:BU-algebra-barreled}
	 The complete Borchers-Uhlmann algebra $\overline{\BU\big(\Gamma_0(\mathfrak{X})\big)}$ is barrelled.
\end{lemma}
\begin{proof}
The spaces $\Gamma_0\big(\mathfrak{X}^{\boxtimes n}\big)$ of compactly supported sections of a complex vector bundle are LF-spaces, as they are defined as the inductive limit of the Frech\'et spaces of sections with support in some compact $K_l$ where $\left\{ K_l \right\}_l$ is a fundamental sequence of compact $K_l \subset \N$ (see \cite[17.2.2 and 17.3.1]{Dieudonne1972}). Since LF-spaces are barrelled \cite[Chapter 33, Corollary 3]{Treves1967} and the direct sum of barrelled spaces is again barrelled \cite[18.11]{KellyNamioka1963}, we find for any $N \in \IN$ that $\BU_N$ is barrelled. Additionally, the inductive limit of barrelled spaces is barrelled \cite[Chapter V, Proposition 6]{RobertsonRobertson1973}, hence the complete BU-algebra over smooth compactly supported sections $\Gamma_0(\mathfrak{X})$ over a complex vector bundle $\mathfrak{X}$ is barrelled.
\end{proof}
We will use barrelled spaces in order to apply the following result:
\begin{lemma}\label{lem:jointcontinuity}
Let $X$ be a barrelled locally convex space, let $\eta:[c,d]\to X$ be a continuous map on a closed interval and let $L_m:X\to Y$ be a family of continuous linear maps into a locally convex space $Y$ indexed by $m\in [a,b]$. If the map $m\mapsto L_m$ is weakly continuous, i.\,e. if $m\mapsto L_mx$ is continuous on $[a,b]$ for each $x\in X$, then the map $(m,m')\mapsto L_m\eta(m')$ is continuous on $[a,b]\times[c,d]$.
\end{lemma}
\begin{proof}
The weak continuity of $m\mapsto L_m$ implies that for each $x\in X$ the image of $m\mapsto L_mx$ is compact. The family of maps $L_m$ is therefore pointwise bounded. Because $X$ is barrelled we may apply the uniform boundedness principle to find that the maps $L_m$ are equicontinuous. For any $(m_0,m'_0)\in[a,b]\times[c,d]$ we set $x \coloneqq \eta(m'_0)$ and we pick an arbitrary convex open neighbourhood $y+V$ of $y \coloneqq L_{m_0}x$, where $V$ is an open neighbourhood of $0$. By equicontinuity there is an open neighbourhood $U\subset X$ of $0$ such that $L_m(U)\subset \frac{1}{2}V$ for all $m\in[a,b]$. As $\eta$ is continuous there is an open neighbourhood $W'\subset[c,d]$ of $m'_0$ such that $\eta(W')\subset x+U$. Similarly there is an open neighbourhood $W\subset[a,b]$ of $m_0$ such that $L_mx-y\in\frac{1}{2}V$ for all $m\in W$. It follows that for all $(m,m')\in W\times W'$
\begin{equation}
L_m\eta(m')-y = L_m(\eta(m')-x) + (L_mx-y)\in \frac{1}{2}V+\frac{1}{2}V\subset V
\end{equation}
which proves the desired continuity.
\end{proof}
For our next lemma we will call an element of $\BU\big(\Dzs \big)$ symmetric if and only if it is totally symmetric in each degree.
\begin{lemma}[Symmetrization of fields] \label{lem:symmetrization-of-fields}
	Let $\BU_S\left(\Dzs \right)$ denote the linear subspace of the Borchers-Uhlmann algebra of initial data $\BU\big(\Dzs \big)$ consisting of symmetric elements. Then there is a unique continuous linear surjective projection $S:\BU\big(\Dzs \big)\to \BU_S\big(\Dzs \big)$ whose kernel is $\KERN{S}=\ICCR$ as defined in Section \ref{sec:CCR}.
\end{lemma}
\begin{proof}
For each $N\ge 1$ and each permutation $\sigma$ of the set $\{1,\ldots,N\}$ we introduce the permutation operator $P^{(N)}_{\sigma}:
\big(\Gamma_0(\TsS\oplus\TsS)\big)^{\otimes N}\to\big(\Gamma_0(\TsS\oplus\TsS)\big)^{\otimes N}$ defined by
\begin{equation}
\left(P^{(N)}_{\sigma}f\right)\big(p_1,\ldots,p_N\big) \coloneqq f\big(p_{\sigma(1)},\ldots,p_{\sigma(N)}\big) \formspace,
\end{equation}
where we view elements of $\big(\Gamma_0(\TsS\oplus\TsS)\big)^{\otimes N}$ as sections in $\Gamma_0\big((\TsS\oplus\TsS)^{\boxtimes N}\big)$. The symmetric tensor product $\big(\Dzs\big)^{\otimes_S N}$ is then the range space of the projection
\begin{equation}
P^{(N)} \coloneqq \frac{1}{N!}\sum\limits_{\sigma}P^{(N)}_{\sigma}.
\end{equation}
Note that each $P^{(N)}_{\sigma}$ is continuous, because the topology of $\Gamma_0\big( (\TsS \oplus \TsS)^{\boxtimes N} \big) $ is invariant under the swapping of variables. It follows that $P^{(N)}$ is a continuous surjection.\par 
We will first argue that $\BU_S\big(\Dzs \big)\cap\ICCR=\{0\}$. For this we note that each $f\in\ICCR$ is of the form
\begin{equation}
f = \sum_{i=1}^{k} h_i \cdot \big( -\i \mathcal{G}_m(\psi_i, \psi'_i) , 0 , \psi_i \otimes \psi'_i - \psi'_i \otimes \psi_i , 0 , 0 ,\dots  \big) \cdot \tilde{h}_i 
\end{equation}
for some $k \in \IN$, $h_i, \tilde{h}_i \in \BU\big(\Dzs\big)$ and $\psi_i, \psi'_i \in \Dzs$, where we have used the shorthand notation $\Gm{\psi_i}{ \psi'_i} = \langle \pi_i , \varphi'_i \rangle_\Sigma - \langle \varphi_i , \pi'_i \rangle_\Sigma$ for $\psi_i = (\varphi_i , \pi_i)$. If $f\not=0$ then its highest degree part is of some degree 
$N\ge2$ and we can write it explicitly, using the above representation, as
\begin{equation}
f^{(N)} = \sum_{i=1}^{k} h_i^{(N_i)}\, \big( \psi_i \otimes \psi'_i - \psi'_i \otimes \psi_i \big)\, \tilde{h}_i^{(N-2-N_i)} \formspace,
\end{equation}
where $h_i^{(N_i)}$ is the highest degree part of $h_i$ and $\tilde{h}_i^{(N-2-N_i)}$ is either the highest degree part of $\tilde{h}_i$ or $0$. It follows by inspection that $P^{(N)}f^{(N)}=0$. Now, if $f\in\ICCR$ is non-zero and symmetric and if $f^{(N)}$ is its highest degree part, then $f^{(N)}=P^{(N)}f^{(N)}=0$, contradicting that $f^{(N)}$ is the highest degree part. It follows that $\BU_S\big(\Dzs \big)\cap\ICCR=\{0\}$.\par 
We now construct for each degree $N\ge2$ two continuous linear maps
\begin{align}
\alpha^{(N)}:&\big(\Gamma_0(\TsS\oplus\TsS)\big)^{\otimes N}\to \big(\Gamma_0(\TsS\oplus\TsS)\big)^{\otimes N}\formspace,\notag\\
\beta^{(N)}:&\big(\Gamma_0(\TsS\oplus\TsS)\big)^{\otimes N}\to \ICCR\formspace,\notag
\end{align}
(for $N=2$ we use $\big(\Gamma_0(\TsS\oplus\TsS)\big)^{\otimes (N-2)}=\IC$) such that
\begin{equation}\label{eqn:symmetrization_of_field}
f = P^{(N)}f + \alpha^{(N)}f + \beta^{(N)}f\formspace.
\end{equation}
We start with the observation that
\begin{equation}\label{eqn:symmetrization_of_field2}
f = P^{(N)}f - \frac{1}{N!} \sum\limits_\sigma  (P^{(N)}_\sigma - 1) f\formspace.
\end{equation}
Every permutation $\sigma$ can be written as a composition $\sigma = \tau_1 \comp \tau_2 \comp \cdots \comp \tau_l$, where each $\tau_i$ is a transposition of neighbouring indices. We then find $P^{(N)}_\sigma = P^{(N)}_{\tau_1}\cdot P^{(N)}_{\tau_2}\cdots P^{(N)}_{\tau_l}$ and, using a telescoping series,
\begin{equation}\label{eqn:commuted-elements}
\big(P^{(N)}_\sigma - 1\big) f = \sum_{i=1}^l \big(P^{(N)}_{\tau_i} - 1\big)\,P^{(N)}_{\tau_{i+1}}\cdots P^{(N)}_{\tau_{l}}\,   f_m^{(N+2)} \formspace.
\end{equation}
This is now a sum over terms where the left-most operator $P^{(N)}_{\tau_i}- 1$ yields a commutator. Using the CCR we may reduce this commutator to a term of lower degree, i.\,e.
\begin{equation}
\big(P^{(N)}_{\tau_i} - 1\big) f' = \tilde{f}'+g \formspace,\quad \tilde{f}'\in \Gamma_0(\Dzs)^{\otimes (N-2)}\formspace,\quad g\in\ICCR
\end{equation}
for any $f'\in \Gamma_0(\Dzs)^{\otimes N}$, where $\tilde{f}'$ depends continuously on $f'$ and hence so does $g$. Repeating this procedure for each term in Equation (\ref{eqn:commuted-elements}) and each term in the sum in Equation (\ref{eqn:symmetrization_of_field2}) yields a well-defined expression of the form
\begin{equation}
f = P^{(N)}f + \sum_j\tilde{f}_j+\sum_jg_j\formspace,
\end{equation}
where $j$ runs over some index set, $\tilde{f}_j$ is homogeneous of degree $N-2$ and $g_j\in\ICCR$. Because $\tilde{f}_j$ and $g_j$ depend continuously on $f$, it suffices to define $\alpha^{(N)}f \coloneqq \sum_j\tilde{f}_j$ and $\beta^{(N)}f \coloneqq \sum_jg_j$. We refer to \cite[Lemma B.5]{Schambach2016} for more details.

In Equation (\ref{eqn:symmetrization_of_field}) we may now proceed to symmetrise the term $\alpha^{(N)}f$ of degree $N-2$. Note that elements of degree 0 or 1 are automatically symmetric. By induction we can then show that
\begin{equation}\label{eqn:fullysymmetrised}
f = \sum\limits_{j=0}^{\lfloor N/2\rfloor}P^{(N-2j)}\alpha^{(N+2-2j)}\cdots \alpha^{(N)}f +
\sum\limits_{j=0}^{\lfloor N/2\rfloor}\beta^{(N-2j)}\alpha^{(N+2-2j)}\cdots \alpha^{(N)}f \formspace.
\end{equation}
(Here the maps $\alpha$ are to be omitted when $j=0$.) We now define $S$ as $S=\bigoplus_{N=0}^{\infty}S_N$ in terms of the continuous linear maps
\begin{align}
S_N:& \big(\Gamma_0(\TsS\oplus\TsS)\big)^{\otimes N}\to \BU_S\left(\Dzs \right) \formspace,\notag\\
&f\mapsto \sum\limits_{j=0}^{N/2}P^{(N-2j)}\alpha^{(N+2-2j)}\cdots \alpha^{(N)}f \formspace,
\end{align}
for all $N\ge 0$. Note that $S$ is continuous and because the $\alpha^{(N)}$ and $\beta^{(N)}$ vanish on symmetric elements, $S$ acts as the identity on $\BU_S\left(\Dzs \right)$. It follows from Equation (\ref{eqn:fullysymmetrised}) that every element $f\in\BU\big(\Dzs \big)$ can be decomposed into $f=f'+g$, where $f'$ is symmetric and $g\in \ICCR$. This decomposition is unique, since $\BU_S\big(\Dzs \big)\cap\ICCR=\{0\}$, and we must have $f'=Sf$. This entails in particular that 
\begin{equation}
\BU\left(\Dzs \right)=\BU_S\big(\Dzs \big)\oplus\ICCR,
\end{equation}
that $\KERN{S}=\ICCR$ and that $S$ is the unique projection with the given range and kernel.
\end{proof}
%
%
%
%

\section{Proof of the energy estimate (\ref{Eqn_energyestimate})}\label{app_energy_estimate}
In this appendix we prove the energy estimate (\ref{Eqn_energyestimate}), which we now restate.
\begin{theorem}
Let $P$ be a normally hyperbolic operator on a real vector bundle $V$ over a globally hyperbolic spacetime $M$ and let $\Sigma\subset M$ be a smooth, space-like Cauchy surface. For all compact sets $K\subset\Sigma$ and $L\subset\mathbb{R}$ there is a $C>0$ such that
\begin{equation}\label{Eqn_energyestimate2}
\int_{D(K)}\|v^{(r)}\|^2\le C\int_K \left(\norm{\restr{v^{(r)}}{\Sigma}}^2+\norm{\restr{n^{\alpha}\nabla_{\alpha}v^{(r)}}{\Sigma}}^2\right)+C\int_{D(K)}\norm{f^{(r)}}^2,
\end{equation}
where $D(k)$ is the domain of dependence and $v^{(r)}$ is a solution to $(P+r)v^{(r)}=f^{(r)}$.
\end{theorem}
\begin{proof}
We may identify $M=\mathbb{R}\times S$ and $g=-Ndt^2+h_t$, where $t\in\mathbb{R}$, $N>0$, $\Sigma_t \coloneqq \{t\}\times S$ is a smooth spacelike Cauchy surface with metric $h_t$ and $\Sigma=\Sigma_0$. We set $\xi_{\alpha} \coloneqq -N\nabla_{\alpha}t$, so that $\xi^{\alpha}$ is a future pointing time-like vector field and $n^{\alpha} \coloneqq N^{-\frac12}\xi^{\alpha}$ is its normalisation. Without loss of generality we may assume that the auxiliary norm $\norm{\PH}$ on $TM$ is given by $2n_{\alpha}n_{\beta}+g_{\alpha\beta}$.

For the purposes of this proof we choose the connection $\nabla$ on $V$ to be the one which is compatible with the auxiliary metric on $V$. Any different choice of connection in (\ref{Eqn_energyestimate2}) can easily be accommodated for by adjusting $C$ at the end of the proof.
Note that for suitable smooth bundle homomorphisms $A$ and $B$ it holds $P=g^{\alpha\beta}\nabla_{\alpha}\nabla_{\beta}+A^{\alpha}\nabla_{\alpha}+B$.

Let us fix $r$ for now and drop the superscripts on $v$ and $f$. We define the quantities
\begin{align}
T_{\alpha\beta}&\coloneqq \nabla_{\alpha}v\cdot\nabla_{\beta}v-\frac12g_{\alpha\beta}\left(\norm{\nabla v}^2+\norm{v}^2\right)\formspace,\\
P_{\alpha}&\coloneqq \xi^{\beta}T_{\alpha\beta}\formspace,\\
\epsilon&\coloneqq n^{\alpha}P_{\alpha}=\sqrt{N}n^{\alpha}n^{\beta}T_{\alpha\beta}\notag\\
&\phantom{:}=\frac12\sqrt{N}\left((2n^{\alpha}n^{\beta}+g^{\alpha\beta})\nabla_{\alpha}v\cdot\nabla_{\beta}v+\norm{v}^2\right)\formspace,
\end{align}
where $\cdot$ refers to the hermitian inner product on $V$. Note that $\epsilon\ge0$.

We may now choose a $T>0$ such that $D(K)\subset (-\infty,T)\times S$ and a compact $K'\subset\Sigma$ which contains $K$ in its interior. Then we may choose an auxiliary Cauchy surface $\Sigma'$ of $(-\infty,T)\times S$ such that $D(K)$ lies to the past of $\Sigma'$, but $\Sigma'$ contains $\Sigma\setminus K'$. Furthermore, we may choose a $C\ge1$ such that the following inequalities hold on $[0,T]\times S$: 
\begin{align}
N^{\pm\frac12}\le C \formspace,\quad
\pm(\nabla^{\alpha}\xi^{\beta}+\nabla^{\beta}\xi^{\alpha})\le C\sqrt{N}(2n^{\alpha}n^{\beta}+g^{\alpha\beta})\formspace,\notag \\
\left\vert\nabla_{\alpha}\xi^{\alpha}\right\vert\le C\sqrt{N}\formspace,\quad
\norm{R}\le C\formspace,\quad
\norm{A}\le C \quad\textrm{and}\quad 
\norm{B}\le C \formspace, 
\end{align}
where $R$ is the curvature of $\nabla$ on $V$. In addition we may assume that $|r+1|\le C$ for all $r\in L$ and that $h_t\le Ch_{t'}$ on $K'$ for all $t,t'\in[0,T]$ and similarly for the hermitian metric in $V$.\par
It will be convenient to introduce $L_t \coloneqq \Sigma_t\cap J^-(\Sigma')$ for $t\in[0,T]$ and the ``energy''
\begin{equation}
\epsilon(t) \coloneqq \int_{L_t}\epsilon\formspace.
\end{equation}
We now want to estimate the quantity
\begin{equation}
E(t) \coloneqq \int_{([0,t]\times S)\cap J^-(\Sigma')}\epsilon
\end{equation}
for $t\in[0,T]$. We note first of all that
\begin{equation}
\frac{d}{dt}E(t)\le \lim_{\tau\to0^+}\tau^{-1}\int_{[t,t+\tau]\times L_t}\epsilon\le C\int_{L_t}\epsilon\formspace,
\end{equation}
where the constant $C$ is needed to estimate the factor $\sqrt{N}$ which arises due to a change of volume form. Furthermore, using Stokes' Theorem:
\begin{equation}
\int_{([0,t]\times S)\cap J^-(\Sigma')}\nabla^{\alpha}P_{\alpha}=\epsilon(t)-\epsilon(0)+\int_{\Sigma'\cap([0,t]\times S)}\nu^{\alpha}P_{\alpha}\formspace,
\end{equation}
where $\nu^{\alpha}$ is the forward unit normal to $\Sigma'$. One may show that the bilinear form $\nu^{\alpha}n^{\beta}+n^{\alpha}\nu^{\beta}-g^{\alpha\beta}n^{\gamma}\nu_{\gamma}$ is positive definite and $n^{\gamma}\nu_{\gamma}<0$. This entails that $\nu^{\alpha}P_{\alpha}\ge0$ and hence
\begin{equation}
\epsilon(t)-\epsilon(0)\le\int_{([0,t]\times S)\cap J^-(\Sigma')}\nabla^{\alpha}P_{\alpha} \formspace.
\end{equation}
Furthermore, we may estimate
\begin{equation}
\left\vert \nabla^{\alpha}P_{\alpha}\right\vert\le \left\vert T_{\alpha\beta}\nabla^{\alpha}\xi^{\beta} \right\vert +\left\vert\xi^{\beta}\nabla^{\alpha}T_{\alpha\beta}\right\vert \formspace,
\end{equation}
where
\begin{equation}
\nabla^{\alpha}T_{\alpha\beta}=v\cdot R_{\alpha\beta}\cdot \nabla^{\alpha}v-v\cdot B\cdot\nabla_{\beta}v-(r+1)v\cdot\nabla_{\beta}v
+f\cdot\nabla_{\beta}v \formspace.
\end{equation}
For the term involving $f$ we can use the further estimate
\begin{equation}
\left\vert\xi^{\beta}f\cdot\nabla_{\beta}v\right\vert\le C \norm{f}\cdot\norm{\nabla v}\le \frac{1}{2} C \left( \norm{f}^2+\norm{\nabla v}^2 \right) \formspace.
\end{equation}
Using our choice of $C$ we can then estimate all the terms in $\nabla^{\alpha}P_{\alpha}$ to find
\begin{equation}
\epsilon(t)\le\epsilon(0)+\int_{([0,t]\times S)\cap J^-(\Sigma')}8 C^2\epsilon+\frac{1}{2} C \norm{f}^2
\end{equation}
and consequently
\begin{equation}
\frac{d}{dt}E(t)\le C\epsilon(t)\le 8C^3E(t)+C\epsilon(0)+\frac{1}{2}C^2\int_{D(K')}\norm{f}^2.
\end{equation}
Therefore, $\frac{d}{dt}\e^{-8C^2t}E(t)\le C\epsilon(0)+\frac{1}{2}C^2\int_{D(K')}\norm{f}^2$. With $E(0)=0$ this yields
\begin{equation}
e^{-8C^2T}E(T)=\int_0^T\frac{d}{dt}\e^{-8C^2t}E(t)dt\le \left(C\epsilon(0)+\frac{1}{2}C^2\int_{D(K')}\norm{f}^2\right)T
\end{equation}
and hence
\begin{equation}
E(T)\le C'\big(\epsilon(0)+\int_{D(K')}\norm{f}^2\big)
\end{equation}
for a suitable $C'>0$ independent of $r$. Note that
$E(T)\ge \int_{D(K)}\norm{v}^2$ and that $\epsilon(0)\le C'\int_K\big(\norm{\restr{v^{(r)}}{\Sigma}}^2+\norm{n^{\alpha}\nabla_{\alpha}\restr{v^{(r)}}{\Sigma}}^2\big)$
when we choose $C'$ large enough. Finally, we may shrink $K'$ to $K$ without adjusting the constants $C$ or $C'$ which leads to the desired estimate.
\end{proof}
%


\noindent
\printbibliography[heading=bibintoc,title={References}]

\end{document}